\documentclass[12pt,cls,onecolumn]{IEEEtran}
\usepackage{graphicx,amsmath,amssymb,epsfig, amsfonts, cite, latexsym, cuted, multicol, multirow, subfigure, stfloats, array, tabularx}
\usepackage{subeqnarray}
\usepackage{color}
\usepackage{setspace}
\usepackage{anysize}

\begin{document}

\title{Opportunistic Downlink Interference Alignment for Multi-Cell MIMO Networks}
\author{\large Hyun Jong Yang, \emph{Member}, \emph{IEEE}, Won-Yong Shin, \emph{Senior Member}, \emph{IEEE}, \\ Bang Chul Jung, \emph{Senior Member}, \emph{IEEE}, Changho Suh, \emph{Member}, \emph{IEEE}, \\and
Arogyaswami Paulraj, \emph{Fellow}, \emph{IEEE} \\
\thanks{H. J. Yang is with the School of Electrical and Computer Engineering, UNIST, Ulsan 689-798, Republic of Korea (E-mail:
hjyang@unist.ac.kr).}
\thanks{W.-Y. Shin is with the Department of Computer Science and
Engineering, Dankook University, Yongin 448-701, Republic of Korea
(E-mail: wyshin@dankook.ac.kr).}
\thanks{B. C. Jung (corresponding author) is with the Department of Electronics Engineering, Chungnam National
University, Daejeon 305-764, Republic of Korea (E-mail:
bcjung@cnu.ac.kr).}
\thanks{C. Suh is with the Department of Electrical Engineering, KAIST, Daejeon 305-701, Republic of Korea (E-mail: chsuh@kaist.ac.kr).}
\thanks{A. Paulraj is with the Department of Electrical Engineering,
Stanford University, Stanford, CA 94305 (email:
apaulraj@stanford.edu). }
} \maketitle


\markboth{IEEE Transactions on Wireless Communications} {Yang {\em
et al.}: Opportunistic Downlink Interference Alignment for
Multi-Cell MIMO Networks}


\newtheorem{definition}{Definition}
\newtheorem{theorem}{Theorem}
\newtheorem{lemma}{Lemma}
\newtheorem{example}{Example}
\newtheorem{corollary}{Corollary}
\newtheorem{proposition}{Proposition}
\newtheorem{conjecture}{Conjecture}
\newtheorem{remark}{Remark}

\def \diag{\operatornamewithlimits{diag}}
\def \min{\operatornamewithlimits{min}}
\def \max{\operatornamewithlimits{max}}
\def \log{\operatorname{log}}
\def \max{\operatorname{max}}
\def \rank{\operatorname{rank}}
\def \out{\operatorname{out}}
\def \exp{\operatorname{exp}}
\def \arg{\operatorname{arg}}
\def \E{\operatorname{E}}
\def \tr{\operatorname{tr}}
\def \SNR{\operatorname{SNR}}
\def \dB{\operatorname{dB}}
\def \ln{\operatorname{ln}}

\def \be {\begin{eqnarray}}
\def \ee {\end{eqnarray}}
\def \ben {\begin{eqnarray*}}
\def \een {\end{eqnarray*}}

\begin{abstract}
In this paper, we propose an opportunistic downlink interference
alignment (ODIA) for interference-limited cellular downlink, which
intelligently combines user scheduling and downlink IA techniques.
The proposed ODIA not only efficiently reduces the effect of
inter-cell interference from other-cell base stations (BSs) but
also eliminates intra-cell interference among spatial streams in
the same cell. We show that the minimum number of users required
to achieve a target degrees-of-freedom (DoF) can be fundamentally
reduced, i.e., the fundamental user scaling law can be improved by
using the ODIA, compared with the existing downlink IA schemes. In
addition, we adopt a limited feedback strategy in the ODIA
framework, and then analyze the number of feedback bits required
for the system with limited feedback to achieve the same user
scaling law of the ODIA as the system with perfect CSI. We also
modify the original ODIA in order to further improve sum-rate,
which achieves the optimal multiuser diversity gain, i.e., $\log
\log N$, per spatial stream even in the presence of downlink
inter-cell interference, where $N$ denotes the number of users in
a cell. Simulation results show that the ODIA significantly
outperforms existing interference management techniques in terms
of sum-rate in realistic cellular environments. Note that the ODIA
operates in a non-collaborative and decoupled manner, i.e., it
requires no information exchange among BSs and no iterative
beamformer optimization between BSs and users, thus leading to an
easier implementation.
\end{abstract}


\newpage


\section{Introduction}
Interference management has been taken into account as one of the
most challenging issues to increase the throughput of cellular
networks serving multiple users. In multiuser cellular
environments, each receiver may suffer from intra-cell and
inter-cell interference.
Interference alignment (IA) was proposed by fundamentally solving
the interference problem when there are multiple communication
pairs~\cite{V_Cadambe08_TIT}. It was shown that the IA scheme can
achieve the optimal degrees-of-freedom (DoF)\footnote{It is
referred that `optimal' DoF is achievable if the outer-bound on
DoF for given network configuration is achievable. } in the
multiuser interference channel with time-varying channel
coefficients. Subsequent studies have shown that the IA is also
useful and indeed achieves the optimal DoF in various wireless
multiuser network setups: multiple-input multiple-output (MIMO)
interference channels~\cite{K_Gomadam11_TIT, T_Gou10_TIT} and
cellular networks~\cite{C_Suh11_TC,C_Suh08_Allerton}. In
particular, IA techniques~\cite{C_Suh11_TC,C_Suh08_Allerton} for
cellular uplink and downlink networks, also known as the
interfering multiple-access channel (IMAC) or interfering
broadcast channel (IBC), respectively, have received much
attention.
 The existing IA framework for cellular networks, however, still has
several practical challenges: the scheme proposed
in~\cite{C_Suh08_Allerton} requires arbitrarily large
frequency/time-domain dimension extension, and the scheme proposed
in~\cite{C_Suh11_TC} is based on iterative optimization of
processing matrices and cannot be optimally extended to an
arbitrary downlink cellular network in terms of achievable DoF.
%

In the literature, there are some results on the usefulness of
fading in single-cell downlink broadcast channels, where one can
obtain multiuser diversity gain along with user scheduling as the
number of users is sufficiently large: opportunistic
scheduling~\cite{R_Knopp95_ICC}, opportunistic
beamforming~\cite{P_Viswanath02_TIT}, and random
beamforming~\cite{M_Sharif05_TIT}. Scenarios exploiting multiuser
diversity gain have been studied also in ad hoc
networks~\cite{W_Shin14_TIT}, cognitive radio
networks~\cite{T_Ban09_TWC}, and cellular
networks~\cite{W_Shin12_TC}.

Recently, the concept of opportunistic IA~(OIA) was introduced
in~\cite{B_Jung11_CL,B_Jung11_TC,H_Yang13_TWC} for the $K$-cell
uplink network (i,e., IMAC model), where there are one $M$-antenna
base station (BS) and $N$ users in each cell. The OIA scheme
incorporates user scheduling into the classical IA framework by
opportunistically selecting $S$ ($S\le M$) users amongst the $N$
users in each cell in the sense that inter-cell interference is
aligned at a pre-defined interference space. It was shown
in~\cite{B_Jung11_TC,H_Yang13_TWC} that one can asymptotically
achieve the optimal DoF if the number of users in a cell scales as
a certain function of the signal-to-noise-ratio (SNR). For the
$K$-cell downlink network (i.e., IBC model) assuming one
$M$-antenna base station (BS) and $N$ per-cell users, studies on
the OIA have been conducted in~\cite{W_Shin12_IEICE,
J_Jose12_Allerton, J_Lee13_TWC, H_Nguyen13_TSP,H_Nguyen13_arXiv,
J_Lee13_arXiv}. More specifically, the user scaling condition for
obtaining the optimal DoF was characterized for the $K$-cell
multiple-input single-output (MISO) IBC~\cite{W_Shin12_IEICE}, and
then such an analysis of the DoF achievability was extended to the
$K$-cell MIMO IBC with $L$ receive antennas at each
user~\cite{J_Jose12_Allerton, J_Lee13_TWC,
H_Nguyen13_TSP,H_Nguyen13_arXiv, J_Lee13_arXiv}---full DoF can be
achieved asymptotically, provided that $N$ scales faster than
${\mathsf{SNR}}^{KM-L}$, for the $K$-cell MIMO IBC using
OIA~\cite{H_Nguyen13_arXiv, J_Lee13_arXiv}.

In this paper, we propose an \textit{opportunistic downlink IA
(ODIA)} framework as a promising interference management technique
for $K$-cell downlink networks, where each cell consists of one BS
with $M$ antennas and $N$ users having $L$ antennas each.
The proposed ODIA jointly takes into account user scheduling and
downlink IA issues. In particular, inspired by the precoder design
in~\cite{C_Suh11_TC}, we use two cascaded beamforming matrices to
construct our precoder at each BS. To design the first transmit
beamforming matrix, we use a user-specific beamforming, which
conducts a linear zero-forcing (ZF) filtering and thus eliminates
intra-cell interference among spatial streams in the same cell. To
design the second transmit beamforming matrix, we use a
predetermined reference beamforming matrix, which plays the same
role of random beamforming for cellular
downlink~\cite{W_Shin12_IEICE, H_Nguyen13_arXiv, J_Lee13_arXiv}
and thus efficiently reduces the effect of inter-cell interference
from other-cell BSs. On the other hand, the receive beamforming
vector is designed at each user in the sense of minimizing the
total amount of received inter-cell interference using
\textit{local} channel state information (CSI) in a decentralized
manner. Each user feeds back both the effective channel vector and
the quantity of received inter-cell interference to its home-cell
BS. The user selection and transmit beamforming at the BSs and the
design of receive beamforming at the users are completely
decoupled. Hence, the ODIA operates in a non-collaborative manner
while requiring no information exchange among BSs and no iterative
optimization between transmitters and receivers, thereby resulting
in an easier implementation.


The main contribution of this paper is four-fold as follows.
\begin{itemize}
\item  We first show that the minimum number of users required to
achieve $S$ DoF ($S\le M$) can be fundamentally reduced to
$\mathsf{SNR}^{(K-1)S-L+1}$ by using the ODIA at the expense of
acquiring perfect CSI at the BSs from users, compared to the
existing downlink IA schemes requiring the user scaling law
$N=\omega(\mathsf{SNR}^{KS-L})$~\cite{H_Nguyen13_arXiv,
J_Lee13_arXiv},\footnote{$f(x) = \omega(g(x))$ implies that
$\lim_{x \rightarrow \infty} \frac{g(x)}{f(x)}=0$.} where $S$
denotes the number of spatial streams per cell. The interference
decaying rate with respect to $N$ for given SNR is also
characterized in regards to the derived user scaling law. \item We
introduce a limited feedback strategy in the ODIA framework, and
then analyze the required number of feedback bits leading to the
same DoF performance as that of the ODIA assuming perfect
feedback, which is given by $\omega\left( \log_2
\mathsf{SNR}\right)$. \item  We present a user scheduling method
for the ODIA to achieve optimal multiuser diversity gain, i.e.,
$\log\log N$ per stream even in the presence of downlink
inter-cell interference. \item To verify the ODIA schemes, we
perform numerical evaluation via computer simulations. Simulation
results show that the proposed ODIA significantly outperforms
existing interference management and user scheduling techniques in
terms of sum-rate in realistic cellular environments.
\end{itemize}

The remainder of this paper is organized as follows. Section
\ref{SEC:system} describes the system and channel models. Section
\ref{SEC:OIA} presents the overall procedure of the proposed ODIA.
In Section \ref{sec:achievability}, the DoF achievablility result
is shown. Section \ref{SEC:OIA_limited} presents the ODIA scheme
with limited feedback. In Section \ref{SEC:Threhold_ODIA}, the
achievability of the spectrally efficient ODIA leading to a better
sum-rate performance is characterized. Numerical results are shown
in Section \ref{SEC:Sim}. Section \ref{SEC:Conc} summarizes the
paper with some concluding remarks.


\section{System and Channel Models} \label{SEC:system}
We consider a $K$-cell MIMO IBC where each cell consists of a BS
with $M$ antennas and $N$ users with $L$ antennas each. The number
of selected users in each cell is denoted by $S (\le M)$. It is
assumed that each selected user receives a single spatial stream.
To consider nontrivial cases, we assume that $L < (K-1)S +1$,
because all inter-cell interference can be completely canceled at
the receivers (i.e., users) otherwise. Moreover, the number of
antennas at the users is in general limited due to the size of the
form factor, and hence it is more safe to assume that $L$ is
relatively small compared to $(K-1)S+1$. The channel matrix from
the $k$-th BS to the $j$-th user in the $i$-th cell is denoted by
$\mathbf{H}_{k}^{[i,j]}\in \mathbb{C}^{L \times M}$, where $i,k\in
\mathcal{K} \triangleq \{ 1, \ldots, K\}$ and $j \in \mathcal{N}
\triangleq \{1, \ldots, N\}$. Each element of
$\mathbf{H}_k^{[i,j]}$ is assumed to be independent and
identically distributed (i.i.d.) according to $\mathcal{CN}(0,1)$.
In addition, quasi-static frequency-flat fading is assumed, i.e.,
channel coefficients are constant during one transmission block
and change to new independent values for every transmission block.
The $j$-th user in the $i$-th cell can estimate the channels
$\mathbf{H}_{k}^{[i,j]}$, $k=1, \ldots, K$, using pilot signals
sent from all the BSs.

The received signal vector at the $j$-th user in the $i$-th cell
is expressed as:
\begin{align} \label{eq:received_y}
\mathbf{y}^{[i,j]} &= \sum_{k=1}^{K}\mathbf{H}_{k}^{[i,j]}
\mathbf{s}_k+  \mathbf{z}^{[i,j]},
\end{align}
where $\mathbf{s}_k\in\mathbb{C}^{M \times 1}$ is the transmit
signal vector at the $k$-th BS with unit average power, i.e., $E
\left\|\mathbf{s}_k\right\|^2 = 1$, and $\mathbf{z}^{[i,j]} \in
\mathbb{C}^{L \times 1}$ denotes  additive noise,  each element of
which is independent and identically distributed complex Gaussian
with zero mean and the variance of $N_0$. The average SNR is given
by $\mathsf{SNR} = {\mathbb{E}\left[\left\|
\mathbf{H}_i^{[i,j]}\mathbf{s}_i \right\|^2
\right]}/{\mathbb{E}\left[ \left\| \mathbf{z}^{[i,j]}\right\|^2
\right]} = {1}/{N_0}$. Thus, in what follows we shall use the
notation $N_0 = \mathsf{{SNR}^{-1}}$ for notational simplicity.

Figure \ref{fig:system_model} shows an example of the MIMO IBC
model, where $K=3$, $M=3$, $S=2$, $L=3$, and $N=2$. The details in
the figure will be described in the subsequent section.

\begin{figure*} \label{fig:system_model}
\begin{center}
  \includegraphics[width=0.87\textwidth]{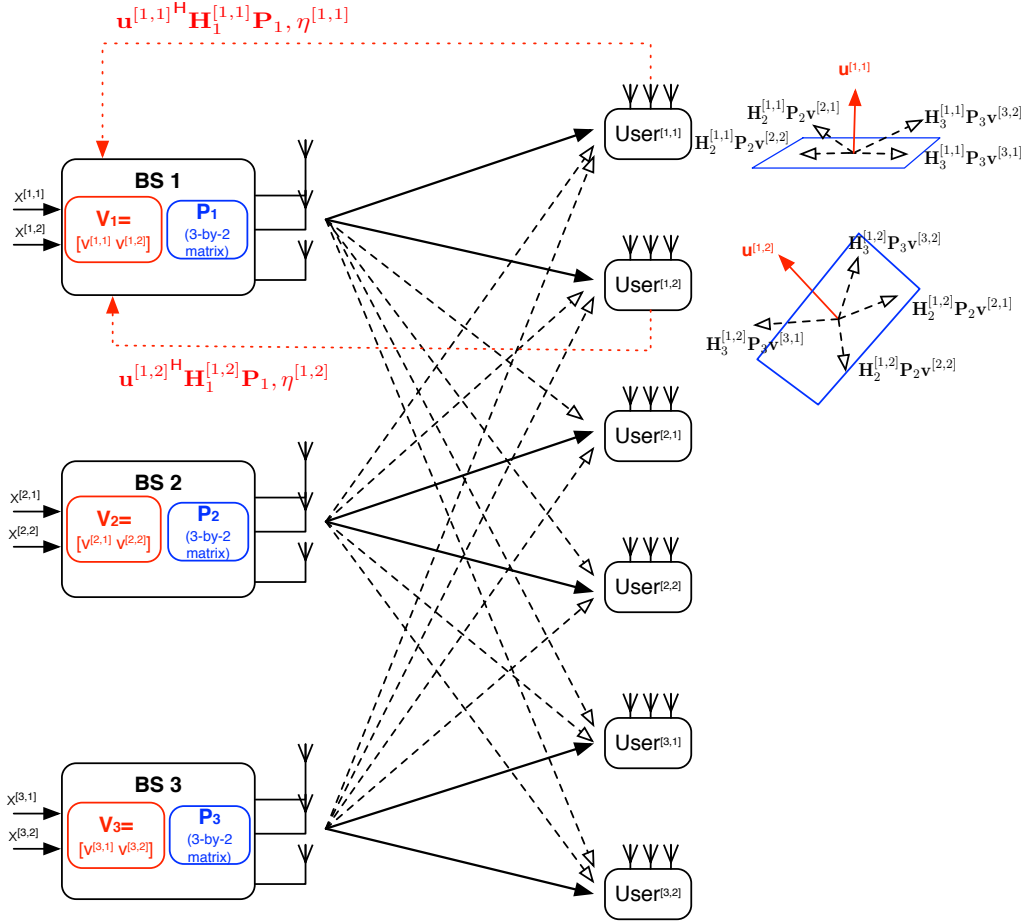}\\
  \caption{The MIMO IBC model, where $K=3$, $M=3$, $S=2$, $L=3$, and $N=2$.}\label{fig:system_model}
  \end{center}
\end{figure*}


\section{Proposed ODIA} \label{SEC:OIA}

 We first describe the overall procedure of our proposed ODIA scheme for the MIMO IBC, and then define its achievable sum-rate and DoF.
\subsection{Overall Procedure} \label{subsec:overall}

The ODIA scheme is described according to the following four
steps.
\subsubsection{Initialization (Broadcast of Reference Beamforming Matrices)}
The reference beamforming matrix at the BS in the $k$-th cell is
given by $\mathbf{P}_k = \left[ \mathbf{p}_{1,k}, \ldots,
\mathbf{p}_{S,k}\right]$, where $\mathbf{p}_{s,k} \in
\mathbb{C}^{M \times 1}$ is an \textcolor{black}{orthonormal
vector} for $k\in \mathcal{K}$ and $s =1, \ldots, S$.
\textcolor{black}{That is, $\mathbf{P}_k$ is an orthonormal basis
for an $S$-dimensional subspace of $\mathbb{C}^{M \times M}$.}
\textcolor{black}{Each BS randomly generates $\mathbf{P}_{k}$
independently of the other BSs.} If the reference beamforming
matrix is generated in a pseudo-random fashion, i.e., it changes
based on a certain pattern as if it changes randomly and the
pattern is known by the BSs as well as the users, BSs do not need
to broadcast them to users.
 Then, the $j$-th user in the $i$-th cell obtains $\mathbf{H}^{[i,j]}_{k}$ and $\mathbf{P}_k$, $k=1, \ldots, K$.
%

\subsubsection{Receive Beamforming  \& Scheduling Metric Feedback}
In the second step, we explain how to decide a user scheduling
metric at each user along with given receive beamforming, where
the design of receive beamforming will be explained in
Section~\ref{sec:achievability}. Let $\mathbf{u}^{[i,j]} \in
\mathbb{C}^{L \times 1}$ denote the unit-norm weight vector at the
$j$-th user in the $i$-th cell, i.e., $\left\| \mathbf{u}^{[i,j]}
\right\|^2 = 1$. Note that the user-specific beamforming
$\mathbf{V}_k$ will be utilized only to cancel intra-cell
interference out, and the inter-cell interference will be
suppressed from user scheduling, which will be specified later.
Thus, from the notion of $\mathbf{P}_k$ and
$\mathbf{H}^{[i,j]}_{k}$, the $j$-th user in the $i$-th cell can
compute the following quantity while using its receive beamforming
vector  $\mathbf{u}^{[i,j]}$, which is given by
\begin{align}\label{eq:eta_tilde}
\tilde{\eta}^{[i,j]}_{k} &= \left\|
{\mathbf{u}^{[i,j]}}^{\mathsf{H}}\mathbf{H}_{k}^{[i,j]}
\mathbf{P}_k \right\|^2,
\end{align}
where $i\in \mathcal{K}$, $j \in \mathcal{N} $, and $k\in
\mathcal{K}\setminus i= \{1, \ldots, i-1, i+1, \ldots, K\}$. Using
(\ref{eq:eta_tilde}), the scheduling metric at the $j$-th user in
the $i$-th cell, denoted by $\eta^{[i,j]}$, is defined as the sum
of $\tilde{\eta}^{[i,j]}_{k}$. That is,
\begin{align} \label{eq:eta}
\eta^{[i,j]} &= \sum_{k=1, k\neq i}^{K} \tilde{\eta}^{[i,j]}_{k}.
\end{align}

As illustrated in Fig. \ref{fig:system_model}, each user feeds the
metric in (\ref{eq:eta}) back to its home-cell BS.
In addition to the scheduling metric in (\ref{eq:eta}), for each
BS to design  the user-specific beamforming $\mathbf{V}_k$, each
user needs to feed back the information of the following vector
\begin{equation} \label{eq:effective_CH}
\mathbf{f}_{i}^{[i,j]} \triangleq
\left({\mathbf{u}^{[i,j]}}^{\mathsf{H}} \mathbf{H}^{[i,j]}_i
\mathbf{P}_ i\right)^{\mathsf{H}}.
\end{equation}

\subsubsection{User Scheduling}
Upon receiving $N$ users' scheduling metrics in the serving cell,
each BS selects $S$ users having the metrics up to the $S$-th
smallest one. Without loss of generality, the indices of selected
users in every cell are assumed to be $(1, \ldots, S)$.
Although $\tilde{\eta}^{[i,j]}_{k}$ is not exactly the amount of
the generating interference from the $k$-th BS to the $j$-th user
in the $i$-th cell due to the absence of $\mathbf{V}_k$, it
decouples the design of the user-specific precoding matrix
$\mathbf{V}_k$ from the user scheduling metric calculation, i.e.,
$\eta^{[i,j]}_{k}$ includes no information of $\mathbf{V}_k$. In
addition, we shall show in the sequel that the inter-cell
interference can be successfully suppressed by using the metric
$\eta^{[i,j]}_{k}$ even with $\mathbf{V}_k$ excluded and that the
optimal DoF can be achieved.

At this point, it is worthwhile to note that the role of
$\mathbf{P}_k$ is two-fold. First, it determines the dimension of
the effective received channel according to given parameter $S$.
By multiplying $\mathbf{P}_k$ to the channel matrix, the dimension
of the effective channel is reduced to $S$ rather than $M$, which
results in reduced number of inter-cell interference terms as well
as reduced average interference level for each interference term.
We shall show in the sequel that $\mathbf{P}_k$ plays a role in
the end of rendering the user scaling law dependent on the
parameter $S$.

Second, $\mathbf{P}_k$ separates the user scheduling procedure
from the user-specific precoding matrix design of $\mathbf{V}_k$
and also from the receiver beamforming vector design of
$\mathbf{u}_k$. By employing the cascaded precoding matrix design,
the scheduling metric in (1) becomes independent of $\mathbf{V}_k$
or $\mathbf{u}_k$, and $\mathbf{u}_k$ can be obtained as a
function of only $\mathbf{H}_k^{[i,j]}$ and $\mathbf{P}_k$ as
shown in (\ref{eq:u_design}).

The reason why $\mathbf{P}_k$ is designed to change in a
pseudo-random fashion is to increase the fairness of the users
scheduling by randomizing the scheduling metric of each user, but
can also be fixed if the fairness is not a matter or the channel
changes fast enough. In addition, if one wants to further improve
the achievable rate, $\mathbf{P}_k$ may be channel-specifically
designed combined with the user scheduling, which however results
in a collaborative and iterative user scheduling and precoding
matrix design.

In this and subsequent sections, we focus on how to simply design
a user scheduling method to guarantee the optimal DoF. An enhanced
scheduling algorithm jointly taking into account the vector to be
fed back in (\ref{eq:effective_CH}) and the scheduling metric  in
(\ref{eq:eta}) may provide a better performance in terms of
sum-rate, which shall be discussed in Section
\ref{SEC:Threhold_ODIA}.



\subsubsection{Transmit Beamforming \& Downlink Data Transmission}
As illustrated in Fig.~\ref{fig:system_model}, the precoding
matrix at each BS is composed of the product of the predetermined
reference beamforming matrix $\mathbf{P}_k$ and the user-specific
precoding matrix $\mathbf{V}_i = \left[ \mathbf{v}^{[i,1]},
\ldots, \mathbf{v}^{[i,S]}\right]$, where $\mathbf{v}^{[i,s]} \in
\mathbb{C}^{S \times 1}$, $i\in \mathcal{K}$. Let us denote the
transmit symbol at the $i$-th BS transmitted to the $j$-th user by
$x^{[i,j]}$, where  $E\left| x^{[i,s]} \right|^2=1/S$ for $s=1,
\ldots, S$. Denoting the transmit symbol vector by $\mathbf{x}_i =
\left[ x^{[i,1]}, \ldots, x^{[i,S]}\right]^T$, the transmit signal
vector at the $i$-th BS is given by $\mathbf{s}_i = \mathbf{P}_i
\mathbf{V}_i \mathbf{x}_i$, and the received signal vector at the
$j$-th user in the $i$-th cell is written as
\begin{align} \label{eq:rec_vector}
\mathbf{y}^{[i,j]} &= \mathbf{H}_i^{[i,j]}\mathbf{P}_i \mathbf{V}_i \mathbf{x}_i + \sum_{k=1, k\neq i}^{K} \mathbf{H}_k^{[i,j]}\mathbf{P}_k \mathbf{V}_k \mathbf{x}_k + \mathbf{z}^{[i,j]}  \nonumber \\
&= \underbrace{\mathbf{H}_i^{[i,j]}\mathbf{P}_i \mathbf{v}^{[i,j]}
x^{[i,j]}}_{\textsf{desired signal}} +  \underbrace{\sum_{s=1,
s\neq j}^{S} \mathbf{H}_i^{[i,j]}\mathbf{P}_i \mathbf{v}^{[i,s]}
x^{[i,s]}}_{\textsf{intra-cell interference}} \nonumber \\ & +
\underbrace{\sum_{k=1, k\neq i}^{K}
\mathbf{H}_k^{[i,j]}\mathbf{P}_k \mathbf{V}_k
\mathbf{x}_k}_{\textsf{inter-cell interference}} +
\mathbf{z}^{[i,j]}.
\end{align}
  The received signal vector after receive beamforming, denoted by $\tilde{y}^{[i,j]} = {\mathbf{u}^{[i,j]}}^{\mathsf{H}} \mathbf{y}^{[i,j]}$, can be rewritten as:
\begin{align}\label{eq:rec_vector_after_BF}
\tilde{y}^{[i,j]}
&= {\mathbf{f}_{i}^{[i,j]}}^{\mathsf{H}} \mathbf{v}^{[i,j]}
x^{[i,j]}  +{\mathbf{f}_{i}^{[i,j]}}^{\mathsf{H}}\sum_{s=1, s\neq
j}^{S} \mathbf{v}^{[i,s]} x^{[i,s]} \nonumber \\  &+ \sum_{k=1,
k\neq i}^{K} {\mathbf{f}_{k}^{[i,j]}}^{\mathsf{H}} \mathbf{V}_k
\mathbf{x}_k +
{\mathbf{u}^{[i,j]}}^{\mathsf{H}}\mathbf{z}^{[i,j]},
\end{align}
where
${\mathbf{f}_{k}^{[i,j]}}^{\mathsf{H}}={\mathbf{u}^{[i,j]}}^{\mathsf{H}}\mathbf{H}_k^{[i,j]}\mathbf{P}_k$.
By selecting users with small $\eta^{[i,j]}$ in (\ref{eq:eta}),
$\mathbf{H}_k^{[i,j]}\mathbf{P}_k$ tends to be orthogonal to the
receive beamforming vector $\mathbf{u}^{[i,j]}$; thus, inter-cell
interference channel matrices
$\mathbf{H}_k^{[i,j]}\mathbf{P}_k\mathbf{V}_k$ in
(\ref{eq:rec_vector_after_BF}) also tend to be orthogonal to
$\mathbf{u}^{[i,j]}$ as illustrated in Fig.
\ref{fig:system_model}.

To cancel out intra-cell interference, the user-specific
beamforming matrix $\mathbf{V}_i \in \mathbb{C}^{S \times S}$is
given by
\begin{align} \label{eq:ZF_BF}
\mathbf{V}_i &= [\mathbf{v}^{[i,1]},\mathbf{v}^{[i,2]}, \ldots,
\mathbf{v}^{[i,S]}] \nonumber \\ &= \begin{bmatrix}
       {\mathbf{u}^{[i,1]}}^{\mathsf{H}} \mathbf{H}_i^{[i,1]} \mathbf{P}_i  \\
       {\mathbf{u}^{[i,2]}}^{\mathsf{H}} \mathbf{H}_i^{[i,2]} \mathbf{P}_i  \\
       \vdots \\
       {\mathbf{u}^{[i,S]}}^{\mathsf{H}} \mathbf{H}_i^{[i,S]} \mathbf{P}_i
     \end{bmatrix}^{-1}   \cdot   \begin{bmatrix}
       \sqrt{\gamma^{[i,1]}} & 0 & \cdots & 0  \\
       0 & \sqrt{\gamma^{[i,2]}} & \cdots & 0  \\
       \vdots & \vdots & \ddots & \vdots \\
       0 & 0 & \cdots & \sqrt{\gamma^{[i,S]}}  \\
     \end{bmatrix},
\end{align}
\label{line:gamma:start}where $\sqrt{\gamma^{[i,j]}}$ denotes a
normalization factor for satisfying the unit-transmit power
constraint for each spatial stream, i.e., $\gamma^{[i,j]} =
1/\left\| \mathbf{P}_i \mathbf{v}^{[i,j]}
\right\|$.\label{line:gamma:end} In consequence, the received
signal can be simplified to
\begin{align}\label{eq:rec_vector_ZF_BF}
\tilde{y}^{[i,j]} &= \sqrt{\gamma^{[i,j]}} x^{[i,j]} +
\underbrace{\sum_{k=1, k\neq i}^{K}
{\mathbf{f}_{k}^{[i,j]}}^{\mathsf{H}}\mathbf{V}_k
\mathbf{x}_k}_{\textsf{inter-cell interference}} +
{\mathbf{u}^{[i,j]}}^{\mathsf{H}}\mathbf{z}^{[i,j]},
\end{align}
which thus does not contain the intra-cell interference term.

As in
\cite{N_Jindal06_TIT,T_Yoo07_JSAC,J_Thukral09_ISIT,R_Krishnamachari10_ISIT,S_Pereira07_Asilomar,B_Jung11_TC},
we assume no loss in exchanging signaling messages such as
information of effective channels, scheduling metrics, and receive
beamforming vectors.

\subsection{Achievable Sum-Rate and DoF}\label{subsec:sum_rate}
From (\ref{eq:rec_vector_ZF_BF}), the achievable rate of the
$j$-th user in the $i$-th cell is given by
\begin{align} \label{eq:data_rate_single_user}
R^{[i,j]}&=\log_2 \left( 1+ \mathsf{SINR}^{[i,j]} \right)
\nonumber \\ &= \log_2 \left( 1+ \frac{ \gamma^{[i,j]} \cdot
|x^{[i,j]}|^2}{\left|{{\mathbf{u}^{[i,j]}}^{\mathsf{H}}}
\mathbf{z}^{[i,j]}\right|^2+\tilde{I}^{[i,j]}} \right) \nonumber \\
& =\log_2 \left( 1+ \frac{ \gamma^{[i,j]} }{\frac{S}{\mathsf{SNR}}
+ \sum_{k=1, k\neq i}^{K} \sum_{s=1}^{S} \left|
{\mathbf{f}_{k}^{[i,j]}}^{\mathsf{H}} \mathbf{v}^{[k,s]}\right|^2
} \right),
\end{align}
where $\tilde{I}^{[i,j]} \triangleq \sum_{k=1, k\neq i}^{K}
\left|{\mathbf{f}_{k}^{[i,j]}}^{\mathsf{H}} \mathbf{V}_k
\mathbf{x}_k\right|^2$.

Using (\ref{eq:data_rate_single_user}), the achievable total DoF
can be defined as \cite{S_Jafar08_TIT}
\begin{equation}\label{eq:sum_DoF}
\textrm{DoF} = \lim_{\textsf{SNR} \rightarrow \infty}
\frac{\sum_{i=1}^{K}\sum_{j=1}^{S}R^{[i,j]}}{\log \textsf{SNR}}.
\end{equation}

\section{DoF Achievability}\label{sec:achievability}
In this section, we characterize the DoF achievability in terms of
the user scaling law with the optimal receive beamforming
technique. To this end, we start with the receive beamforming
design that maximizes the achievable DoF. For given channel
instance, from (\ref{eq:data_rate_single_user}), each user can
attain the maximum DoF of 1 if and only if the interference
$\sum_{k=1, k\neq i}^{K} \sum_{s=1}^{S}
\Big|{\mathbf{f}_{k}^{[i,j]}}^{\mathsf{H}}\mathbf{v}^{[k,s]}\Big|^2
\cdot \mathsf{SNR}$ remains constant for increasing SNR. Note that
$R^{[i,j]}$ can be bounded as
\begin{align}
\!\!\! &R^{[i,j]} \!\! \ge \!\! \log_2 \!\! \left(\!\!1+ \frac{ \gamma^{[i,j]} }{ \frac{S}{\mathsf{SNR}}+ \sum_{k=1, k\neq i}^{K} \sum_{s=1}^{S} \left\| \mathbf{f}_{k}^{[i,j]}\right\|^2 \left\| \mathbf{v}^{[k,s]}\right\|^2 } \!\!\! \right) \label{eq:data_rate_single_user_bound} \\
& \ge \log_2 \left( 1+ \frac{ \gamma^{[i,j]} }{ \frac{S}{\mathsf{SNR}}+ \sum_{k\neq i}^{K} \sum_{s=1}^{S} \left\| \mathbf{f}_{k}^{[i,j]} \right\|^2 \left\| \mathbf{v}^{(\max)}_{i}\right\|^2 } \right) \label{eq:data_rate_single_user_bound2}\\
& = \log_2\left( \mathsf{SNR}\right) + \log_2
\left(\frac{1}{\mathsf{SNR}}+ \frac{ \frac{\gamma^{[i,j]}}{\left\|
\mathbf{v}^{(\max)}_{i}\right\|^2}}{ \frac{S}{\left\|
\mathbf{v}^{(\max)}_{i}\right\|^2}+ I^{[i,j]} } \right),
\label{eq:data_rate_single_user_bound3}
\end{align}
where $\mathbf{v}^{(\max)}_{i}$ in
(\ref{eq:data_rate_single_user_bound2}) is defined by
\begin{align} \label{eq:v_max_def}
\mathbf{v}^{(\max)}_{i} &= \arg \max\bigg\{ \left\|
\mathbf{v}^{[i',j']}\right\|^2: i'\in \mathcal{K}\setminus i,
j'\in \mathcal{S}\bigg\},
\end{align}
$\mathcal{S} \triangleq \{1, \ldots, S\}$, and $I^{[i,j]}$ in
(\ref{eq:data_rate_single_user_bound3}) is defined by
\begin{align}
I^{[i,j]} \triangleq \sum_{k=1, k\neq i}^{K} \sum_{s=1}^{S}\left\|
\mathbf{f}_{k}^{[i,j]}\right\|^2 \cdot \mathsf{SNR}.
\end{align}
Here, $\mathbf{v}_i^{(\max)}$ is fixed for given channel instance,
because $\mathbf{v}^{[i,j]}$ is determined by
$\mathbf{H}_i^{[i,j]}$, $j=1, \ldots, S$. Recalling that the
indices of the selected users are $(1, \ldots, S)$ for all cells,
we can expect the DoF of 1 for each user if and only if for some
$0 \le \epsilon< \infty$,
\begin{equation}
I^{[i,j]} < \epsilon, \hspace{10pt} \forall j \in \mathcal{S},
i\in \mathcal{K}.
\end{equation}


To maximize the achievable DoF, we aim to minimize the
sum-interference $\sum_{i=1}^{K} \sum_{j=1}^{S}I^{[i,j]}$ through
receive beamforming at the users. Since $I^{[i,j]} =
\sum_{s=1}^{S} \eta^{[i,j]} \mathsf{SNR}$, we have
\begin{equation} \label{eq:sum_interference_equiv}
\sum_{i=1}^{K} \sum_{j=1}^{S}I^{[i,j]} = \sum_{i=1}^{K}
\sum_{j=1}^{S}\sum_{s=1}^{S} \eta^{[i,j]} \mathsf{SNR} =
S\sum_{i=1}^{K} \sum_{j=1}^{S} \eta^{[i,j]}\mathsf{SNR}.
\end{equation}
The equation (\ref{eq:sum_interference_equiv}) implies that the
collection of distributed effort to minimize $\eta^{[i,j]}$ at the
users can reduce the sum of received interference. Therefore, each
user finds the beamforming vector that minimizes $\eta^{[i,j]}$
from
\begin{align} \label{eq:u_design}
\mathbf{u}^{[i,j]} &= \arg \min_{\mathbf{u}} \eta^{[i,j]} =  \arg
\min_{\mathbf{u}} \sum_{k=1, k\neq i}^{K} \left\|
\mathbf{u}^{\mathsf{H}}\mathbf{H}_{k}^{[i,j]} \mathbf{P}_k
\right\|^2.
\end{align}
Let us denote the augmented interference matrix by
\begin{align} \label{eq:G_def}
\mathbf{G}^{[i,j]} &\triangleq \Bigg[
\left(\mathbf{H}_{1}^{[i,j]}\mathbf{P}_1\right), \ldots,
\left(\mathbf{H}_{i-1}^{[i,j]}\mathbf{P}_{i-1}\right),
\left(\mathbf{H}_{i+1}^{[i,j]}\mathbf{P}_{i+1}\right), \nonumber
\\ & \quad \quad \quad \ldots,
\left(\mathbf{H}_{K}^{[i,j]}\mathbf{P}_{K}\right)\Bigg]^{\mathsf{H}}
\in \mathbb{C}^{(K-1)S \times L},
\end{align}
and the singular value decomposition of $\mathbf{G}^{[i,j]}$ by
\begin{equation} \label{eq:G_SVD}
\mathbf{G}^{[i,j]} =
\boldsymbol{\Omega}^{[i,j]}\boldsymbol{\Sigma}^{[i,j]}{\mathbf{Q}^{[i,j]}}^{\mathsf{H}},
\displaybreak[0]
\end{equation}
where $\boldsymbol{\Omega}^{[i,j]}\in \mathbb{C}^{(K-1)S\times L}$
and $\mathbf{Q}^{[i,j]}\in \mathbb{C}^{L\times L}$ consist of $L$
orthonormal columns, and $\boldsymbol{\Sigma}^{[i,j]} =
\textrm{diag}\left( \sigma^{[i,j]}_{1}, \ldots,
\sigma^{[i,j]}_{L}\right)$, where $\sigma^{[i,j]}_{1}\ge \cdots
\ge\sigma^{[i,j]}_{L}$. \pagebreak[0] Then, the optimal
$\mathbf{u}^{[i,j]}$ is determined as
\begin{equation} \label{eq:W_SVD}
\mathbf{u}^{[i,j]}= \mathbf{q}^{[i,j]}_{L},
\end{equation}
where $\mathbf{q}^{[i,j]}_{L}$ is the $L$-th column of
$\mathbf{Q}^{[i,j]}$. With this choice the scheduling metric is
simplified to
\begin{equation} \label{eq:LIF_beamforming_simple}
\eta^{[i,j]} = {\sigma^{[i,j]}_{L}}^2.
\end{equation}
Since each column of $\mathbf{P}_k$ is isotropically and
independently distributed, each element of the effective
interference channel matrix $\mathbf{G}^{[i,j]}$ is i.i.d. complex
Gaussian with zero mean and unit variance.

\begin{remark}\label{remark:decoupled}
In general, the conventional scheduling metric such as SNR or SINR
in the IBC is dependent on the precoding matrices at the
transmitters, which makes the joint optimization of the precoder
design and user scheduling difficult to be separated from each
other and implemented with feasible signaling overhead and low
complexity. The previous schemes \cite{Q_Shi11_TSP,
K_Gomadam11_TIT} for the IBC only consider the design of the
precoding matrices and receive filters without any consideration
of user scheduling.

With the cascaded precoding matrix design, however, the proposed
scheme decouples the user scheduling metric calculation and the
user-specific precoding matrix $\mathbf{V}_i$, as shown in
(\ref{eq:eta_tilde}). In addition, the receive beamforming vector
design can also be decoupled from $\mathbf{V}_i$ as shown in
(\ref{eq:u_design}). A similar cascaded precoding matrix design
was used in \cite{C_Suh11_TC} for some particular cases of the
antenna configuration without the consideration of user
scheduling. However, the proposed scheme applies to an arbitrary
antenna and channel configuration, where the inter-cell
interference is suppressed with the aid of opportunistic user
scheduling. In addition, we shall show in the sequel that the
optimal DoF can be achievable under a certain user scaling
condition for an arbitrary antenna configuration without any
iterative optimization procedure between the users and BSs.
\end{remark}
\label{line:remark1:end}

\begin{remark}\label{remark:noiteration}
Note that although it is assumed in the proposed scheme that each
user feeds back the $(1 \times S)$-dimensional vector
$\mathbf{f}_i^{[i,j]}$ to its home cell, the amount of CSI
feedback is equivalent to that in the conventional single-cell
MU-MIMO scheme such as ZF or minimum mean-squared error (MMSE)
precoding. On the other hand, the previous iterative transceiver
design schemes \cite{Q_Shi11_TSP, K_Gomadam11_TIT} based on local
CSI for the IBC require all the selected users to feed back the
information of the receive beamformer to all the BSs in the
network, which results in $K$ times more feedback compared to the
single-cell MU-MIMO scheme even for one iteration where the users
feed back their receive beamformers and the BSs update their
transmit precoders once. Furthermore, the information of weight
coefficients also needs to be fed back to all the BSs in
\cite{Q_Shi11_TSP}. We shall show via numerical simulations in the
sequel that even with $K$ times less feedback the proposed scheme
exhibits superior sum-rate compared to the iterative scheme
\cite{Q_Shi11_TSP}.
\end{remark}
We start with the following lemma to derive the achievable DoF.
\begin{lemma}[Lemma 1 \cite{H_Yang13_TWC}] \label{lemma:F_phi}  \label{line:lemma:F_phi}
The CDF of $\eta^{[i,j]}$, denoted by $F_{\eta}(x)$, can be
written as
\begin{equation} \label{eq:F_phi}
F_{\eta}(x) = c_0 x^{(K-1)S-L+1} + o\left(x^{(K-1)S-L+1} \right),
\end{equation}
for $0 \le x <1$, where $f(x)=o(g(x))$ means $\lim_{x\rightarrow
\infty} \frac{f(x)}{g(x)} = 0$, and $\tilde{c}_0$ is a constant
determined by $K$, $S$, and $L$.
\end{lemma}


We further present the following lemma for the probabilistic
interference level of the ODIA.
\begin{lemma} \label{lemma:CDF_scaling}
The sum-interference remains constant with high probability for
increasing SNR, that is,
\begin{align}
\label{eq:P_def}\mathcal{P}&\triangleq
\lim_{\textsf{SNR}\rightarrow \infty} \textrm{Pr}
\Bigg\{\sum_{i=1}^{K}\sum_{j=1}^{S} I^{[i,j]} \le \epsilon
\Bigg\}=1 \displaybreak[0]
\end{align}
for any $0<\epsilon<\infty$,  if
\begin{equation}
N = \omega\left( \mathsf{SNR}^{(K-1)S-L+1} \right).
\end{equation}
\end{lemma}
\begin{proof}
See appendix \ref{app:lemma2}.
\end{proof}

Now, the following theorem establishes the DoF achievability of
the proposed ODIA.

\begin{theorem}[User scaling law] \label{theorem:DoF}
The proposed ODIA scheme with the scheduling metric
(\ref{eq:LIF_beamforming_simple}) achieves the optimal $KS$ DoF
for given $S$ with high probability  if
\begin{equation} \label{eq:N_scaling}
N=\omega\left(\textsf{SNR}^{(K-1)S-L+1}\right).
\end{equation}
\end{theorem}

\begin{proof}
If the sum-interference remains constant for increasing SNR with
probability $\mathcal{P}$, the achievable rate in
(\ref{eq:data_rate_single_user_bound3}) can be further bounded by
\begin{align}
&R^{[i,j]}  \nonumber \\ & \ge  \mathcal{P} \!  \cdot \! \left[
\log_2\left( \mathsf{SNR}\right)\!\! + \!\! \log_2 \!\! \left(
\!\! \frac{1}{\mathsf{SNR}} \! + \! \frac{
\gamma^{[i,j]}/\left(S\left\|
\mathbf{v}^{(\max)}_{i}\right\|^2\right)}{ 1/\left\|
\mathbf{v}^{(\max)}_{i}\right\|^2+ \epsilon } \! \right) \!\!
\right], \label{eq:data_rate_single_user_bound5}
\end{align}
for any $0\le \epsilon < \infty$. Thus, the achievable DoF in
(\ref{eq:sum_DoF}) can be bounded by
\begin{equation} \label{eq:DoF_SVD_LB}
\textrm{DoF} \ge KS \cdot \mathcal{P}.
\end{equation}
From Lemma \ref{lemma:CDF_scaling}, it is immediate to show that
$\mathcal{P}$ tends to 1, and hence $KS$ DoF is achievable if $N =
\omega\left( \mathsf{SNR}^{(K-1)S-L+1}\right)$, which proves the
theorem.
\end{proof}

\textcolor{black}{From Theorem \ref{theorem:DoF}, it is shown that
there exist a fundamental trade-off between the achievable DoF
$KS$ and required user scaling of  $N = \omega\left(
\mathsf{SNR}^{(K-1)S-L+1}\right)$. This trade-off can also be
observed in terms of the sum-rate even under a practical system
setup, as we shall show in Section \ref{SEC:Sim}. Therefore, a
higher $S$ value can be chosen to achieve higher DoF or sum-rate
if there exist more users in the network. }

The following remark discusses the uplink and downlink duality on
the DoF achievability within the OIA framework.
\begin{remark}[Uplink-downlink duality on the DoF achievability] \label{remark:up_down_duality}
\label{line:duality_power:start}The same scaling condition of
$N=\omega\left( \mathsf{SNR}^{K(S-1)-L+1}\right)$ was achieved to
obtain $KS$ DoF in the $K$-cell uplink interference channel
\cite{H_Yang13_TWC}, each cell of which is composed of a BS with
$M$ antennas and $N$ users each with $L$ antennas. Similarly as in
the proposed scheme, the uplink scheme \cite{H_Yang13_TWC} also
selects $S$ users that generate the minimal interference to the
receivers (BSs).  In the uplink scheme, the transmitters (users)
perform SVD-based beamforming and the receivers (BSs) employ a ZF
equalization, while in the proposed downlink case the transmitters
(BSs) perform ZF precoding and the receivers (users) employ
SVD-based beamforming. In addition, each transmitter sends the
information on effective channel vectors to the corresponding
receiver in the uplink case, and vise versa in the downlink case.
The transmit power per spatial stream is the same for both the
cases. Therefore, Theorem \ref{theorem:DoF} implies that the same
DoF is achievable with the same user scaling law for the downlink
and uplink cases. \label{line:duality_power:end}
\end{remark}

The user scaling law characterizes the trade-off between the
asymptotic DoF and number of users, i.e., the more number of
users, the more achievable DoF.
In addition, we relate the derived user scaling law to the interference decaying rate with respect to $N$ for given SNR in the following theorem. 

\begin{theorem}[Interference decaying rate] \label{theorem:scaling_decay}
If the user scaling condition to achieve a target DoF is given by
$N = \omega \left( \textsf{SNR}^{\tau'}\right)$ for some
$\tau'>0$, then the interference decaying rate is given by
\begin{align}
 E\left\{\frac{1}{\eta^{[i,j]}} \right\} \ge \Theta\left( N^{1/\tau'}\right),
\end{align}
where $f(x) = \Theta(g(x))$ if $f(x) = O(g(x))$ and $g(x) =
O(f(x))$.
\end{theorem}
\begin{proof}
\label{line:theorem2_proof:start}From the proof of Theorem
\ref{theorem:DoF}, the user scaling condition to achieve a target
DoF is given by $N = \omega \left( \textsf{SNR}^{\tau'}\right)$ if
and only if the CDF of $\eta^{[i,j]}$ is given by $a_0 x^{-\tau'}
+ o(x^{-\tau'})$ for $\tau'>0$. The theorem can be shown by
following the footsteps of the proof of \cite[Lemma
4]{H_Yang14_TSP}, and the detailed proof is omitted.
\label{line:theorem2_proof:end}
\end{proof}

%
From Theorem \ref{theorem:scaling_decay}, the interference
decaying rate of the proposed ODIA for the $j$th selected user in
the $i$-th cell with respect to $N$ is given by
\begin{equation}
E\left\{\frac{1}{\eta^{[i,j]}} \right\} \ge \Theta\left(
N^{1/((K-1)S-L+1)}\right),
\end{equation}
which is also the same as the result in the uplink channel
\cite{H_Yang14_TSP}. The user scaling law also provides an insight
on the interference decaying rate with respect to $N$ for given
SNR; that is, the smaller SNR exponent of the user scaling law,
the faster interference decreasing rate with respect to $N$.

\vspace{20pt}
%
%
%
%
%
\subsection{Comparison to the previous results}\label{subsec:DoF_comparison}
\label{line:comparison1:start}In this subsection, the DoF
achievability is compared with the previous results in
\cite{H_Nguyen13_arXiv, J_Lee13_TWC, W_Shin12_IEICE}. From
\cite[Lemma 4.2]{H_Nguyen13_arXiv}, choosing $M_i = S$ ($S\le M$)
therein, where $M_i$ denotes the number of spatial streams in the
$i$-th cell, $S$ DoF is achievable per cell, i.e., $KS$ DoF in
total, if $N = \Theta\left( \textrm{SNR}^{\rho}\right)$ for
$\rho>KS-L$; or equivalently,
\begin{align} \label{eq:N_scaling_previous}
N = \omega\left( \textrm{SNR}^{KS-L}\right).
\end{align}
In addition, from \cite[Theorem 6]{J_Lee13_TWC}, choosing $d = S$
($S\le M$) therein, which is the target DoF for each cell, $KS$
DoF is achievable, under the same scaling condition given in
(\ref{eq:N_scaling_previous}). The same conclusion was obtained in
\cite{W_Shin12_IEICE}. Intuitively, the exponent of SNR in the
user scaling condition represents the number of interference
spatial streams after suppression and nulling. Note that the
number of total interference spatial streams received at each user
is $KS-1$ excluding one desired spatial stream, and that the
receive diversity for nulling received interference is $L-1$
leaving one spatial domain for receiving a desired stream. Thus,
the exponent of SNR becomes $(KS-1)-(L-1)=KS-L$ as shown in
(\ref{eq:N_scaling_previous}).

On the other hand, the proposed ODIA pre-nulls $S-1$ intra-cell
interference signals at the transmitter, and hence the exponent
becomes $(KS-1)-(S-1)-(L-1) = (K-1)S-L+1$ as shown in Theorem
\ref{theorem:DoF}. This improvement in the user scaling condition
is attributed to the additional CSI feedback of
${\mathbf{u}^{[i,j]}}^{\mathsf{H}}\mathbf{H}_i^{[i,j]}
\mathbf{P}_i$, which are used to design the precoding matrix
$\mathbf{V}_i$ in (\ref{eq:ZF_BF}). This feedback procedure
corresponds to the feedforward of the effective channel vectors in
the uplink OIA case \cite{H_Yang13_TWC}.

Note that even with this feedback procedure, a straightforward
dual transceiver and user scheduling scheme inspired by the uplink
OIA would result in an infinitely-iterative optimization between
the user scheduling and transceiver design, because the received
interference changes according to the precoding matrix and receive
beamforming vector. Furthermore, only with the cascaded precoding
matrix, the iterative optimization is still needed, since the
coupled optimization issue is still there, as shown in
\cite{C_Suh11_TC}.
 It is indeed the proposed ODIA that can achieve  the same user scaling condition of the uplink case, i.e., $N = \omega\left( \textrm{SNR}^{(K-1)S-L+1}\right)$, without any iterative design. In addition, the
proposed ODIA applies to an arbitrary $M$, $L$, and $K$, whereas
the optimal DoF is achievable only in a few special cases in the
scheme proposed in \cite{C_Suh11_TC}.



\section{ODIA with Limited feedback} \label{SEC:OIA_limited}
In the proposed ODIA scheme, the vectors
(${\mathbf{u}^{[i,j]}}^{\mathsf{H}}\mathbf{H}^{[i,j]}_{i}\mathbf{P}_i$)
in (\ref{eq:effective_CH}) can be fed back to the corresponding BS
using pilots rotated by the effective channels
\cite{L_Choi04_TWC}. However, this analog feedback requires two
consecutive pilot phases for each user: regular pilot for uplink
channel estimation and analog feedback for effective channel
estimation. Hence, pilot overhead grows with respect to the number
of users in the network. As a result, in practical systems with
massive users, it is more preferable to follow the widely-used
limited feedback approach \cite{D_Love03_TIT}, in which  the
information of
${\mathbf{u}^{[i,j]}}^{\mathsf{H}}\mathbf{H}^{[i,j]}_{i}\mathbf{P}_i$
is fed back using  codebooks.

For limited feedback, we define the codebook by
\begin{equation}
\mathcal{C}_f = \left\{ \mathbf{c}_{1}, \ldots,
\mathbf{c}_{N_f}\right\},
\end{equation}
where $N_f$ is the codebook size and $\mathbf{c}_k\in
\mathbb{C}^{S \times 1}$ is a unit-norm codeword, i.e.,
$\|\mathbf{c}_i\|^2=1$. Hence, the number of feedback bits used is
given by
\begin{equation}
n_f = \lceil\log_2 N_f \rceil  (\textrm{bits})
\end{equation}
For ${\mathbf{f}_{i}^{[i,j]}}^{\mathsf{H}} =
{\mathbf{u}^{[i,j]}}^{\mathsf{H}} \mathbf{H}^{[i,j]}_{i}
\mathbf{P}_i,$ each user quantizes the normalized vector for given
$\mathcal{C}_f$ from
\begin{align} \label{eq:f_tilde}
\tilde{\mathbf{f}}_{i}^{[i,j]} = \arg \max_{ \{\mathbf{w} =
\mathbf{c}_k: 1\le k \le N_f\}} \frac{\left|
{\mathbf{f}_{i}^{[i,j]}}^{\mathsf{H}}\mathbf{w}\right|^2}{\left\|
\mathbf{f}_{i}^{[i,j]}\right\|^2}.
\end{align}
Now, the user feeds back three types of information: 1) index of
$\tilde{\mathbf{f}}_{i}^{[i,j]}$, 2) channel gain of $\left\|
\mathbf{f}_{i}^{[i,j]}\right\|^2$, and 3) scheduling metric
$\eta^{[i,j]}$. Note that the feedback of scalar information such
as channel gains and scheduling metrics can be fed back relatively
accurately with a few bits of uplink data, and the main challenge
is on the feedback of the angle of vectors \cite{D_Love03_TIT}.
Thus, in what follows, the aim is to analyze the impact of the
quantized feedback of the index of
$\tilde{\mathbf{f}}_{i}^{[i,j]}$.
 Then, BS $i$ constructs the quantized vectors $\hat{\mathbf{f}}^{[i,j]}$ from
\begin{align}\label{eq:f_hat_def}
\hat{\mathbf{f}}^{[i,j]} \triangleq \left\|
\mathbf{f}_{i}^{[i,j]}\right\|^2 \cdot
\tilde{\mathbf{f}}_{i}^{[i,j]}, \hspace{10pt}i=1, \ldots, S,
\end{align}
and the precoding matrix $\hat{\mathbf{V}}_i$ from
\begin{align} \label{eq:V_hat}
\hat{\mathbf{V}}_i = \hat{\mathbf{F}}_i^{-1}
\boldsymbol{\Gamma}_i,
\end{align}
where $\boldsymbol{\Gamma}_i = \textrm{diag} \left(
\sqrt{\gamma^{[i,1]}}, \ldots, \sqrt{\gamma^{[i,S]}}\right)$ and
$\hat{\mathbf{F}}_i = \left[ \hat{\mathbf{f}}^{[i,1]}, \ldots,
\hat{\mathbf{f}}^{[i,S]}\right]^{\mathsf{H}}$.
%

With limited feedback, the received signal vector after receive
beamforming is written by
\begin{align}\label{eq:rec_vector_after_BF_limited}
\tilde{y}^{[i,j]} &= {\mathbf{f}_{i}^{[i,j]}}^{\mathsf{H}}\hat{\mathbf{V}}_i \mathbf{x}_i + \cdot \sum_{k=1, k\neq i}^{K} {\mathbf{f}_{k}^{[i,j]}}^{\mathsf{H}} \hat{\mathbf{V}}_k \mathbf{x}_k + {\mathbf{u}^{[i,j]}}^{\mathsf{H}} \mathbf{z}^{[i,j]}  \\
& = \sqrt{\gamma^{[i,j]}}x^{[i,j]} + \underbrace{\left(
{\mathbf{f}_{i}^{[i,j]}}^{\mathsf{H}}\hat{\mathbf{V}}_i
\mathbf{x}_i-
\sqrt{\gamma^{[i,j]}}x^{[i,j]}\right)}_{\textrm{residual
intra-cell interference}} \nonumber \\ &+ \sum_{k=1, k\neq i}^{K}
{\mathbf{f}_{k}^{[i,j]}}^{\mathsf{H}} \hat{\mathbf{V}}_k
\mathbf{x}_k + {\mathbf{u}^{[i,j]}}^{\mathsf{H}}
\mathbf{z}^{[i,j]},
\end{align}
where the residual intra-cell interference is non-zero due to the
quantization error in $\hat{\mathbf{V}}_i$.

It is important to note that the residual intra-cell interference
is a function of $\hat{\mathbf{V}}_i$, which includes other users'
channel information, and thus each user treats this term as
unpredictable noise and calculates only the inter-cell
interference for the scheduling metric as in (\ref{eq:eta}); that
is, the scheduling metric is not changed for the ODIA with limited
feedback.

The following theorem establishes the user scaling law for the
ODIA with limited feedback.
\begin{theorem} \label{th:codebook}
The ODIA with a
Grassmannian\footnote{\label{line:footnote_Grassmannian}The
Grassmannian codebook refers to a vector codebook having a
maximized minimum chordal distance of any two codewords, which can
be obtained by solving the Grassmannian line packing problem
\cite{D_Love03_TIT}.} or random codebook achieves the same user
scaling law of the ODIA with perfect CSI described in Theorem
\ref{theorem:DoF}, if
\begin{equation} \label{eq:nf_cond0}
n_f =\omega\left( \log_2 \mathsf{SNR} \right).
\end{equation}
That is, $KS$ DoF is achievable with high probability if $N=\omega
\left( \mathsf{SNR}^{(K-1)S-L+1}\right)$ and (\ref{eq:nf_cond0})
holds true.
\end{theorem}
\begin{proof}
Without loss of generality, the quantized vector
$\hat{\mathbf{f}}^{[i,j]}$ can be decomposed as
\begin{align} \label{eq:f_hat2}
\hat{\mathbf{f}}^{[i,j]} &= \left\|
\mathbf{f}_{i}^{[i,j]}\right\|^2 \cdot
\tilde{\mathbf{f}}_{i}^{[i,j]} \nonumber \\ & =
\sqrt{1-{d^{[i,j]}}^2}\cdot\mathbf{f}_{i}^{[i,j]}+
d^{[i,j]}\left\|\mathbf{f}_{i}^{[i,j]}\right\|^2
\left(\mathbf{t}^{[i,j]}\right),
\end{align}
where $\mathbf{t}^{[i,j]}$ is a unit-norm vector i.i.d. over
$\textrm{null}\left( \mathbf{f}_{i}^{[i,j]}\right)$
\cite{N_Jindal06_TIT, Z_Peng16_TWC}. At this point, we consider
the worse performance case where each user finds
$\hat{\mathbf{f}}^{[i,j]}$ such that with a slight abuse of
notation
\begin{align} \label{eq:f_hat3}
\hat{\mathbf{f}}^{[i,j]} =
\sqrt{1-{d^{\max}_i}^2}\cdot\mathbf{f}_{i}^{[i,j]}+
d^{\max}_i\nu_i \cdot\mathbf{t}^{[i,j]},
\end{align}
where
\begin{eqnarray}
d^{\max}_i &=& \max \left\{d^{[i,1]}, \ldots, d^{[i,S]} \right\},  \nonumber \\
\nu_i &=& \max \left\{ \left\|\mathbf{f}_{i}^{[i,j]}\right\|^2,
j=1, \ldots, S \right\}.
\end{eqnarray}
Note that more quantization error only degrades the achievable
rate, and hence the quantization via (\ref{eq:f_hat3}) yields a
performance lower-bound. Inserting (\ref{eq:f_hat3}) to
(\ref{eq:V_hat}) gives us
\begin{align} \label{eq:V_hat2}
\hat{\mathbf{V}}_i = \left( \sqrt{1- {d^{\max}_i}^2}\mathbf{F}_i +
d^{\max}_i\nu_i\mathbf{T}_i \right)^{-1} \boldsymbol{\Gamma}_i,
\end{align}
where ${\mathbf{F}}_i = \left[ \mathbf{f}_{i}^{[i,1]}, \ldots,
\mathbf{f}_{i}^{[i,S]}\right]^{\mathsf{H}}$ and ${\mathbf{T}}_i =
\left[ {\mathbf{t}}^{[i,1]}, \ldots,
{\mathbf{t}}^{[i,S]}\right]^{\mathsf{H}}$.

The Taylor expansion of $\left(
\sqrt{1-{d^{\max}_i}^2}\mathbf{F}_i + d^{\max}_i\nu_i\mathbf{T}
\right)^{-1}$ in (\ref{eq:V_hat}) gives us
\begin{align} \label{eq:Taylor}
&\left( \sqrt{1-{d^{\max}_i}^2}\mathbf{F}_i +
d^{\max}_i\nu_i\mathbf{T}_i \right)^{-1}  \nonumber \\ &=
\mathbf{F}_i^{-1} - \mathbf{F}_i^{-1}\mathbf{T}_i
\mathbf{F}_i^{-1} \nu_id^{\max}_i +\sum_{k=2}^{\infty}
\mathbf{A}_k \left(d^{\max}_i\right)^k,
\end{align}
where $\mathbf{A}_k$ is a function of $\mathbf{F}_i$ and
$\mathbf{T}_i$. Thus, $\hat{\mathbf{V}}_i$ can be written by
\begin{align} \label{eq:V_hat_final}
\hat{\mathbf{V}}_i = \mathbf{F}_i^{-1}\boldsymbol{\Gamma}_i -
d^{\max}_i\nu_i\mathbf{F}_i^{-1}\mathbf{T}_i
\mathbf{F}_i^{-1}\boldsymbol{\Gamma}_i +\sum_{k=2}^{\infty}
\left(d^{\max}_i\right)^k\mathbf{A}_k\boldsymbol{\Gamma}_i
\end{align}

Inserting (\ref{eq:V_hat_final}) to
(\ref{eq:rec_vector_after_BF_limited}) yields
\begin{align}\label{eq:rec_vector_after_BF_limited3}
\tilde{y}^{[i,j]} &= \sqrt{\gamma^{[i,j]}}x^{[i,j]}  \nonumber \\ &- \underbrace{d^{\max}_i\nu_i{\mathbf{t}^{[i,j]}}^{\textsf{H}}\mathbf{F}_i^{-1}\boldsymbol{\Gamma}_i \mathbf{x}_i + \sum_{k=2}^{\infty}  \left(d^{\max}_i\right)^k{\mathbf{f}_{i}^{[i,j]}}^{\mathsf{H}}\mathbf{A}_k \boldsymbol{\Gamma}_i \mathbf{x}_i }_{\textrm{residual intra-cell interference}} \nonumber \\
& \hspace{0pt}+ \sum_{k=1, k\neq i}^{K}
{\mathbf{f}_{k}^{[i,j]}}^{\mathsf{H}} \hat{\mathbf{V}}_k
\mathbf{x}_k + {\mathbf{u}^{[i,j]}}^{\mathsf{H}}
\mathbf{z}^{[i,j]}.
\end{align}
Consequently,  the rate $R^{[i,j]}$ in
(\ref{eq:data_rate_single_user}) is given by
\begin{align} \label{eq:data_rate_single_user3}
R^{[i,j]}= \log_2 \left( 1+ \frac{ \gamma^{[i,j]} }{ \frac{S+
\Delta^{[i,j]}}{\mathsf{SNR}}+ \sum_{k\neq i}^{K} \sum_{s=1}^{S}
\left| {\mathbf{f}_{k}^{[i,j]}}^{\mathsf{H}}
\mathbf{v}^{[k,s]}\right|^2 } \right),
\end{align}
where
\begin{align} \label{eq:Delta}
\Delta^{[i,j]} = \left(d^{\max}_i\right)^2\delta_1\cdot
\mathsf{SNR} + \sum_{k=2}^{\infty}
\left(d^{\max}_i\right)^{2k}\delta_k\cdot \mathsf{SNR},
\end{align}
\begin{eqnarray}
\delta_1 &=& \left(\nu_i^2{\mathbf{t}^{[i,j]}}^{\textsf{H}}\mathbf{F}_i^{-1}{\boldsymbol{\Gamma}_i}^2\mathbf{F}_i^{-\mathsf{H}}\mathbf{t}^{[i,j]}\right), \nonumber \\
\delta_k &=&
\left({\mathbf{f}_{i}^{[i,j]}}^{\mathsf{H}}\mathbf{A}_k
\boldsymbol{\Gamma}_i^2\mathbf{A}_k^{\mathsf{H}}
\mathbf{f}_{i}^{[i,j]}\right).
\end{eqnarray}
As in (\ref{eq:data_rate_single_user_bound}) to
(\ref{eq:data_rate_single_user_bound3}), the achievable rate can
be bounded by
\begin{align} \label{eq:R_bound_limitedFB}
R^{[i,j]} \!\! & \ge \mathcal{P}'  \! \cdot \!  \left[ \log_2
\mathsf{SNR} + \log_2 \left( \frac{1}{\mathsf{SNR}} \! + \! \frac{
\frac{\gamma^{[i,j]}}{\left\| \mathbf{v}^{(\max)}_{i}\right\|^2}}{
\frac{1}{\left\| \mathbf{v}^{(\max)}_{i}\right\|^2}+ 2\epsilon }
\right) \right],
\end{align}
where
\begin{align}
\label{eq:P_def2}
\mathcal{P}' & \triangleq  \textrm{Pr} \Bigg\{\left(\sum_{i=1}^{K}\sum_{j=1}^{S} I^{[i,j]} \le \epsilon\right) \& \left(\Delta^{[i,j]}/\left\| \mathbf{v}^{(\max)}_{i}\right\|^2 \le \epsilon\right), \nonumber \\ & \quad \quad \quad  \forall i\in \mathcal{K}, j\in \mathcal{S}  \Bigg\} \displaybreak[0] \\
\label{eq:P_def3}& = \textrm{Pr}
\Bigg\{\sum_{i=1}^{K}\sum_{j=1}^{S} I^{[i,j]} \le \epsilon,
\forall i\in \mathcal{K}, j\in \mathcal{S}  \Bigg\}  \nonumber  \\
& \quad \cdot \textrm{Pr} \Bigg\{\Delta^{[i,j]} \le \epsilon',
\forall i\in \mathcal{K}, j\in \mathcal{S}  \Bigg\},
\end{align}
where $\epsilon' \triangleq \epsilon\cdot\left\|
\mathbf{v}^{(\max)}_{i}\right\|^2$.
 Here, (\ref{eq:P_def3}) follows from the fact that the inter-cell interference $I^{[i,j]}$ and residual intra-cell interference  $\Delta^{[i,j]}$ are independent each other. Note also that the level of residual intra-cell interference does not affect the user selection and is determined only by the codebook size $N_f$. Hence, the user selection result does not change for different $N_f$.

The achievable DoF is given by
 \begin{align}
 \textrm{DoF} \ge \lim_{\textsf{SNR} \rightarrow \infty}
 KS \cdot \mathcal{P}'.
 \end{align}
 If $N=\omega\left(\textsf{SNR}^{(K-1)S-L+1}\right)$, the first term of (\ref{eq:P_def3}) tends to 1 according to Theorem \ref{theorem:DoF}.
Thus, the maximum DoF can be obtained if and only if
$\Delta^{[i,j]} \le \epsilon'$ for all selected users for
increasing SNR.

In Appendix \ref{app:th_codebook}, it is shown that
$\Delta^{[i,j]} \le \epsilon'$ for all selected users if $n_f
=\omega\left( \log_2 \mathsf{SNR}\right)$ for both Grassmannian
and random codebooks. Therefore, if
$N=\omega\left(\textsf{SNR}^{(K-1)S-L+1}\right)$ and $n_f
=\omega\left( \log_2 \mathsf{SNR}\right)$, $\mathcal{P}'$ in
(\ref{eq:P_def3}) tends to 1, which proves the theorem.
\end{proof}
From Theorem \ref{th:codebook}, the minimum number of feedback
bits $n_f$ is characterized to achieve the optimal $KS$ DoF, which
increases with respect to $\log_2(\mathsf{SNR})$. It is worthwhile
to note that the results are the same for the Grassmannian and
random codebooks.

We conclude this section by providing the following comparison to
the well-known conventional results on limited feedback systems.

\begin{remark} \label{line:remark_FB1:start}
In the previous works on limited feedback systems, the performance
analysis was focused on the average SNR or the average rate loss
\cite{C_Au-Yeung09_TWC}. In an average sense, the Grassmannian
codebook is in general outperforms the random codebook. However,
our scheme focuses on the asymptotic codebook performance for
given channel instance for increasing SNR, and it turned out that
this asymptotic behaviour is the same for the two codebooks. In
fact, this result agrees with the previous works e.g.,
\cite{B_Khoshnevis11_Thesis}, in which the performance gap between
the two codebooks was shown to be negligible as $n_f$ increases
through computer simulations.  \label{line:remark_FB1:end}
\end{remark}

\begin{remark}
For the MIMO broadcast channel with limited feedback, where the
transmitter has $L$ antennas and employs the random codebook, it
was shown \cite{N_Jindal06_TIT} that the achievable rate loss for
each user, denoted by $\Delta R$, due to the finite size of the
codebook is \textcolor{black}{upper-bounded} by
\begin{equation}
\Delta < \log_2 \left(1+\textrm{SNR} \cdot 2^{-n_f/(L-1)} \right).
\end{equation}
Thus, to achieve the maximum 1 DoF for each user, or to make the
rate loss negligible as the SNR increases, the term $\textrm{SNR}
\cdot 2^{-n_f/(L-1)}$ should remain constant for increasing SNR.
That is, $n_f$ should scale faster than $(L-1)\log_2
(\textrm{SNR})$. \label{line:remark_FB2:start} Note however that
the proof of Theorem \ref{th:codebook} is different from that in
\cite{N_Jindal06_TIT}, since the residual interference due to the
limited feedback, $\Delta^{[i,j]}$, needs to vanish for any given
channel instance with respect to SNR to achieve a non-zero DoF per
spatial stream. \label{line:remark_FB2:end} Though the system and
proof are different, our results of Theorem \ref{th:codebook} are
consistent with this previous result.
\end{remark}

\section{Spectrally Efficient ODIA (SE-ODIA)} \label{SEC:Threhold_ODIA}
In this section, we propose a spectrally efficient OIA (SE-ODIA)
scheme and show that the proposed SE-ODIA achieves the optimal
multiuser diversity gain $\log \log N$. For the DoF achievability,
it was enough to design the user scheduling in the sense to
minimize inter-cell interference. However, to achieve optimal
multiuser diversity gain, the gain of desired channels also needs
to be considered in user scheduling. The overall procedure of the
SE-ODIA follows that of the ODIA described in Section
\ref{SEC:OIA} except the the third stage `User Scheduling'. In
addition, we assume the perfect feedback of the effective desired
channels
${\mathbf{u}^{[i,j]}}^{\mathsf{H}}\mathbf{H}_i^{[i,j]}\mathbf{P}_i$
for the SE-ODIA. We incorporate the semiorthogonal user selection
algorithm proposed in \cite{T_Yoo06_JSAC} to the ODIA framework
taking into consideration inter-cell interference. Specifically,
the algorithm for the user scheduling at the BS side is as
follows:
\begin{itemize}
\item Step 1: Initialization:
\begin{align}
\mathcal{N}_1& = \{1, \ldots, N\}, \hspace{10pt} s=1
\end{align}
\item Step 2: For each user $j\in \mathcal{N}_s$ in the $i$-th
cell, the $s$-th orthogonal projection vector, denoted by
$\tilde{\mathbf{b}}_{s}^{[i,j]}$, for given $\left\{
\mathbf{b}_{1}^{[i]}, \ldots, \mathbf{b}_{s-1}^{[i]} \right\}$ is
calculated from:
\begin{align}
\tilde{\mathbf{b}}^{[i,j]}_s &= \mathbf{f}_{i}^{[i,j]} -
\sum_{s'=0}^{s-1} \frac{{\mathbf{b}_{s'}^{[i]}}^{\mathsf{H}}
\mathbf{f}_{i}^{[i,j]}}{\|\mathbf{b}_{s'}^{[i]}\|^2}\mathbf{b}_{s'}^{[i]}
\end{align}
Note that if $s=1$, $\tilde{\mathbf{b}}_{1}^{[i,j]} =
\mathbf{f}_{i}^{[i,j]}$. \item Step 3: For the $s$-th user
selection, a user is selected at random from the user pool
$\mathcal{N}_s$ that satisfies the following two conditions:
\begin{align}
\label{eq:C}\mathsf{C}_1:& \eta^{[i,j]} \le \eta_I,
\hspace{10pt}\mathsf{C}_2: \|\tilde{\mathbf{b}}_{s}^{[i,j]}\|^2
\ge \eta_D
\end{align}
Denote the index of the selected user by $\pi(s)$ and define
\begin{equation}
\mathbf{b}_{s}^{[i]} = \tilde{\mathbf{b}}^{[i,\pi(s)]}_s.
\end{equation}
\item Step 4: If $s < S$, then find the $(s+1)$-th user pool
$\mathcal{N}_{s+1}$ from:
\begin{align}
\mathcal{N}_{s+1}& = \left\{j:j \in \mathcal{N}_{s}, j \neq \pi(s), \frac{\left|{\mathbf{f}_{i}^{[i,j]}}^{\mathsf{H}} \mathbf{b}_{s}^{[i]}\right|}{\| \mathbf{f}_{i}^{[i,j]}\| \|\mathbf{b}_{s}^{[i]}\|}  <\alpha\right\}, \nonumber \\
 s &= s+1,
\end{align}
where $\alpha>0$ is a positive constant. Repeat Step 2 to Step 4
until $s=S$.
\end{itemize}

To show the SE-ODIA achieves the optimal multiuser diversity gain,
we start with the following lemma for the bound on
$|\mathcal{N}_s|$.
\begin{lemma}\label{lemma:N_card}
The cardinality of $\mathcal{N}_s$ can be bounded by
\begin{align}
|\mathcal{N}_s| & \gtrsim N \cdot \alpha^{2(S-1)}.
\end{align}
The approximated inequality becomes tight as $N$ increases.
\end{lemma}
\begin{proof}
See Appendix \ref{app:Ns_cardinality}.
\end{proof}

We also introduce the following useful lemma.
\begin{lemma}\label{lemma:quadratic}
If $x \in \mathbb{C}^{M \times 1}$ has its element i.i.d.
according to $\mathcal{CN}(0, \sigma^2)$ and $\mathbf{A}$ is an
idempotent matrix of rank $r$ (i.e., $\mathbf{A}^2= \mathbf{A}$),
then $\mathbf{x}^{\mathsf{H}} \mathbf{A} \mathbf{x}/\sigma^2$ has
a Chi-squared distribution with $2r$ degrees-of-freedom.
\end{lemma}
\begin{proof}
See \cite{G_Seber03_Book}.
\end{proof}

In addition, the following lemma on the achievable rate of the
SE-ODIA will be used to show the achievability of optimal
multiuser diversity gain.
\begin{lemma}\label{lemma:effective_gain}
For the $j$-th selected user in the $i$-th cell, the achievable
rate is bounded by
\begin{align} \label{eq:data_rate_single_user4}
R^{[i,j]}\ge \log_2 \left( 1+ \frac{ \frac{\left\|
\mathbf{b}_{j}^{[i]} \right\|^2}{1+ \frac{(S-1)^4
\alpha^2}{1-(S-1)\alpha^2}} }{ \frac{S}{\mathsf{SNR}} +
\sum_{k\neq i}^{K} \sum_{s=1}^{S} \left|
{\mathbf{f}_{k}^{[i,j]}}^{\mathsf{H}}\mathbf{v}^{[k,s]}\right|^2 }
\right).
\end{align}
\end{lemma}
\begin{proof}
Since the chosen channel vectors are not perfectly orthogonal,
there is degradation in the effective channel gain
$\gamma^{[i,j]}$. Specifically, for the $j$-th selected user in
the $i$-th cell, we have
\begin{align}\label{eq:gamma_j_bound}
\gamma^{[i,j]} &= \frac{1}{\left[
\left(\mathbf{F}_i\mathbf{F}_i^{\mathsf{H}}\right)^{-1}
\right]_{j,j}} > \frac{\left\| \mathbf{b}_{j}^{[i]} \right\|^2}{1+
\frac{(S-1)^4 \alpha^2}{1-(S-1)\alpha^2}},
\end{align}
which follows from \cite[Lemma 2]{T_Yoo06_JSAC}. Inserting
(\ref{eq:gamma_j_bound}) to the sum-rate lower bound in
(\ref{eq:data_rate_single_user}) proves the lemma.
\end{proof}

Now the following theorem establishes the achievability of the
optimal multiuser diversity gain.

\begin{theorem}\label{theorem:MUD}
The proposed SE-ODIA scheme with
\begin{equation} \label{eq:eta_D_choice}
\eta_D = \epsilon_D \log \mathsf{SNR}
\end{equation}
\begin{equation}\label{eq:eta_I_choice}
\eta_I = \epsilon_I\mathsf{SNR}^{-1}
\end{equation}
for any $\epsilon_D, \epsilon_I >0$ achieves the optimal multiuser
diversity gain given by
\begin{equation}
R^{[i,j]} = \Theta\left( \log \left( \mathsf{SNR}\cdot \log
N\right)\right),
\end{equation}
with high probability for all selected users in the high SNR
regime if
\begin{align} \label{eq:N_scaling_MUD}
N = \omega \left(
\mathsf{SNR}^{\frac{(K-1)S-L+1}{1-(\epsilon_D/2)}}\right).
\end{align}
\end{theorem}

\begin{proof}
Amongst $|\mathcal{N}_s|$ users, there should exist at least one
user satisfying the conditions $\mathsf{C}_1$ and $\mathsf{C}_2$
to make the proposed user scheduling for the SE-ODIA valid. Thus,
we first show the probability that there exist at least one valid
user, denoted by $\mathsf{p}_s$, converges to 1, for the $s$-th
user selection, if $N$ scales according to
(\ref{eq:N_scaling_MUD}) with the choices (\ref{eq:eta_D_choice})
and (\ref{eq:eta_I_choice}).

The probability that each user satisfies the two conditions is
given by $\textrm{Pr} \left\{ \mathsf{C}_1  \right\} \cdot
\textrm{Pr} \left\{ \mathsf{C}_2  \right\}$, because the two
conditions are independent of each other. Consequently,
$\mathsf{p}_{s}$ is given by
\begin{align}
\mathsf{p}_s &= 1- \left( 1- \textrm{Pr} \left\{ \mathsf{C}_1  \right\} \cdot \textrm{Pr} \left\{ \mathsf{C}_2  \right\} \right)^{|\mathcal{N}_s|} \\
\label{eq:p_s2}& \gtrsim 1- \left( 1- \textrm{Pr} \left\{
\mathsf{C}_1  \right\} \cdot \textrm{Pr} \left\{ \mathsf{C}_2
\right\} \right)^{N \cdot \alpha^{2(S-1)}}.
\end{align}
Note that each element of ${\mathbf{f}_{i}^{[i,j]}}^{\mathsf{H}} =
{\mathbf{u}^{[i,j]}}^{\mathsf{H}} \mathbf{H}^{[i,j]}_{i}
\mathbf{P}_i$ is i.i.d. according to $\mathcal{CN}(0,1)$, because
\textcolor{black}{$\mathbf{P}_k$ is independently and randomly
chosen orthonormal basis for an $S$-dimensional subspace of
$\mathbb{C}^{M \times M}$} and because
${\mathbf{u}^{[i,j]}}^{\mathsf{H}}$ is designed independently of
$\mathbf{H}_i^{[i,j]}$ and isotropically distributed over a unit
sphere. Thus, ${\mathbf{f}_{i}^{[i,j]}}^{\mathsf{H}} =
{\mathbf{u}^{[i,j]}}^{\mathsf{H}} \mathbf{H}^{[i,j]}_{i}
\mathbf{P}_i$ has its element i.i.d. according to
$\mathcal{CN}(0,1)$.

Let us define $\mathbf{P}$ by
\begin{align}
\mathbf{P} &\triangleq \left( \mathbf{I} - \sum_{s'=0}^{s-1}
\frac{ \mathbf{b}_{s'}^{[i]}{\mathbf{b}_{s'}^{[i]}}^{\mathsf{H}}}{
\|\mathbf{b}_{s'}^{[i]}\|^2}\right),
\end{align}
which is a symmetric idempotent matrix with rank $(S-s+1)$. Since
$\mathbf{b}_{s}^{[i]} = \mathbf{P}\mathbf{f}_{i}^{[i,j]}$, from
Lemma \ref{lemma:quadratic},
$\left\|\mathbf{b}_{s}^{[i]}\right\|^2$ is a Chi-squared random
variable with $2(S-s+1)$ degrees-of-freedom.

In Appendix \ref{app:MUD}, for $\eta_D>2$, we show that
\begin{equation} \label{eq:ps_conv}
\lim_{\mathsf{SNR}\rightarrow \infty} \mathsf{p}_s = 1,
\hspace{10pt} \textrm{if } N = \omega \left(
\mathsf{SNR}^{\frac{(K-1)S-L+1}{1-(\epsilon_D/2)}}\right).
\end{equation}

Now, given that there always exist at least one user that
satisfies the conditions $\mathsf{C}_1$ and $\mathsf{C}_2$, the
achievable sum-rate can be bounded from Lemma
\ref{lemma:effective_gain} by
\begin{align} \label{eq:data_rate_single_user5}
R^{[i,j]}&\ge \log_2 \left( 1+ \frac{ \frac{\left\| \mathbf{b}_{j}^{[i]} \right\|^2}{1+ \frac{(S-1)^4 \alpha^2}{1-(S-1)\alpha^2}} \cdot \frac{1}{\left\|\mathbf{v}_i^{\max} \right\|^2}}{ \frac{S}{\mathsf{SNR}\left\|\mathbf{v}_i^{\max} \right\|^2}+ \sum_{k\neq i}^{K} \sum_{s=1}^{S} \left\| \mathbf{f}_{k}^{[i,j]}\right\|^2 } \right)\\
\label{eq:data_rate_single_user6}& \ge \log_2 \left( 1+ \frac{ \frac{\left\| \mathbf{b}_{j}^{[i]} \right\|^2}{1+ \frac{(S-1)^4 \alpha^2}{1-(S-1)\alpha^2}} \cdot \mathsf{SNR}/\left\|\mathbf{v}_i^{\max} \right\|^2}{ S/\left\|\mathbf{v}_i^{\max} \right\|^2+ KS\epsilon_I} \right)\\
\label{eq:data_rate_single_user7}& = \log_2 \left( 1+ \left\| \mathbf{b}_{j}^{[i]} \right\|^2\mathsf{SNR} \cdot \xi \right)\\
\label{eq:data_rate_single_user8}& \ge \log_2 \left( 1+
\epsilon_D(\log N)\cdot \mathsf{SNR} \right),
\end{align}
where (\ref{eq:data_rate_single_user6}) follows from  the fact
that the sum-interference for all selected users, given by
$\sum_{j=1}^{S}\sum_{i=1}^{K} \eta^{[i,j]}\mathsf{SNR}$ (See
(\ref{eq:sum_interference_equiv})), does not exceed $KS \epsilon_I
$ by choosing $\eta_I = \epsilon_I \mathsf{SNR}^{-1}$.
Furthermore, $\xi$ is a constant given by
\begin{equation}
\xi = \frac{1}{\left\| \mathbf{v}_i^{\max}\right\|^2 \left( 1+
\frac{(S-1)^4
\alpha^2}{1-(S-1)\alpha^2}\right)\left(S/\left\|\mathbf{v}_i^{\max}
\right\|^2+ KS\epsilon_I\right)},
\end{equation}
and (\ref{eq:data_rate_single_user8}) follows from $\|
\mathbf{b}_{j}^{[i]} \|^2 \ge \eta_D  = \epsilon_D \log N$.
Therefore, the proposed SE-ODIA achieves the optimal multiuser
diversity gain $\log\log N$ in the high SNR regime, if $N = \omega
\left( \mathsf{SNR}^{\frac{(K-1)S-L+1}{1-(\epsilon_D/2)}}\right)$.
\end{proof}

Therefore, the optimal multiuser gain of $\log\log N$ is achieved
using the proposed SE-ODIA with the choices of
(\ref{eq:eta_D_choice}) and (\ref{eq:eta_I_choice}).
\label{line:SE_ODIA:start} Note that since small $\epsilon_D$
suffices to obtain the optimal multiuser gain, the condition on
$N$ does not dramatically change compared with that required to
achieve $KS$ DoF (See Theorem \ref{theorem:DoF}). Thus,
surprisingly, this means a slight increase in user scaling results
in optimal multiuser diversity by using the proposed SE-ODIA.
\label{line:SE_ODIA:end} Combining the results in Theorem
\ref{theorem:DoF} and \ref{theorem:MUD}, we can conclude the
achievability of the optimal DoF and multiuser gain as follows.
\begin{remark}
In fact, the ODIA described in Section \ref{SEC:OIA} can be
implemented using the SE-ODIA approach by choosing $\eta_D = 0$,
$\alpha = 1$, and $\eta_I^{[i]} = \min\left\{ \eta^{[i,1]},
\ldots, \eta^{[i,N]}\right\}$, where $\eta_I^{[i]}$ denotes
$\eta_I$ at the $i$-th cell. In summary, the optimal $KM$ DoF and
optimal multiuser gain of $\log \log N$ can be achieved using the
proposed ODIA framework, if the number of users per cell increases
according to $N = \omega\left(
\mathsf{SNR}^{\frac{(K-1)M-L+1}{1-(\epsilon_D/2)}}\right)$ for any
$\epsilon_D>0$.
\end{remark}

\section{Numerical Results} \label{SEC:Sim}

In this section, we compare the performance of the proposed ODIA
with two conventional schemes which also utilize the multi-cell
random beamforming technique at BSs. First, we consider ``max-SNR"
technique, in which each user designs the receive beamforming
vector in the sense to maximize the desired signal power, and
feeds back the maximized signal power to the corresponding BS.
Each BS selects $S$ users who have higher received signal power.
Second, ``min-INR" technique is considered, in which each user
performs receive beamforming in order to minimize the sum of
inter-cell interference and intra-cell
interference\cite{H_Nguyen13_arXiv, J_Lee13_arXiv}. Hence,
intra-cell interference does not vanish at users, while the
proposed ODIA perfectly eliminates it via transmit beamforming.
Specifically, from (\ref{eq:rec_vector_after_BF}), the $j$-th user
in the $i$-th cell should calculate the following $S$ scheduling
metrics
 \begin{align}
\eta^{[i,j]}_{\textrm{min-INR}, m} &=
\underbrace{\left\|{\mathbf{u}^{[i,j],m}}^{\mathsf{H}}
\mathbf{H}_i^{[i,j]}\tilde{\mathbf{P}}_{i,m}
\right\|^2}_{\textrm{intra-cell interference}} \nonumber  \\ &+
\underbrace{\sum_{k=1, k\neq
i}^{K}\left\|{\mathbf{u}^{[i,j],m}}^{\mathsf{H}}
\mathbf{H}_k^{[i,j]}\mathbf{P}_{k} \right\|^2}_{\textrm{inter-cell
interference}}, \,\, m=1, \ldots, S,
\end{align}
where $\tilde{\mathbf{P}}_{i,m} \triangleq \left[ \mathbf{p}_{1,
i}, \ldots, \mathbf{p}_{m-1,i}, \mathbf{p}_{m+1, i}, \ldots,
\mathbf{p}_{S,i} \right]$. For each $m$, the receive beamforming
vector $\mathbf{u}^{[i,j],m}$ is assumed to be designed such that
$\eta^{[i,j]}_{\textrm{min-INR}, m}$ is minimized. Each user
feedbacks $S$ scheduling metrics to the corresponding BS, and the
BS selects the user having the minimum scheduling metric for the
$m$-th spatial stream, $m=1, \ldots, S$. For more details about
the min-INR scheme, refer to \cite{H_Nguyen13_arXiv,
J_Lee13_arXiv}. \\

\label{line:Fig_Int:start} Fig. \ref{fig:Interf_N} shows the
sum-interference at all users for varying number of users per
cell, $N$, when $K=3$, $M=4$, $L=2$, and SNR=$20$dB. The solid
lines are obtained from Theorem \ref{theorem:scaling_decay} with
proper biases, and thus only the slopes of the solid lines are
relevant. The decaying rates of sum-interference of the proposed
ODIA are higher than those of the min-INR scheme since intra-cell
interference is perfectly eliminated in the proposed ODIA. In
addition, the interference decaying rates of the proposed ODIA are
consistent with the theoretical results of Theorem
\ref{theorem:scaling_decay}, which proves that the user scaling
condition derived in Theorem \ref{theorem:DoF} and the
interference bound in Theorem \ref{theorem:scaling_decay} are in
fact accurate and tight. \label{line:Fig_Int:end}

\begin{figure}[t]
\begin{center}
  \includegraphics[width=0.64\textwidth, angle=-0]{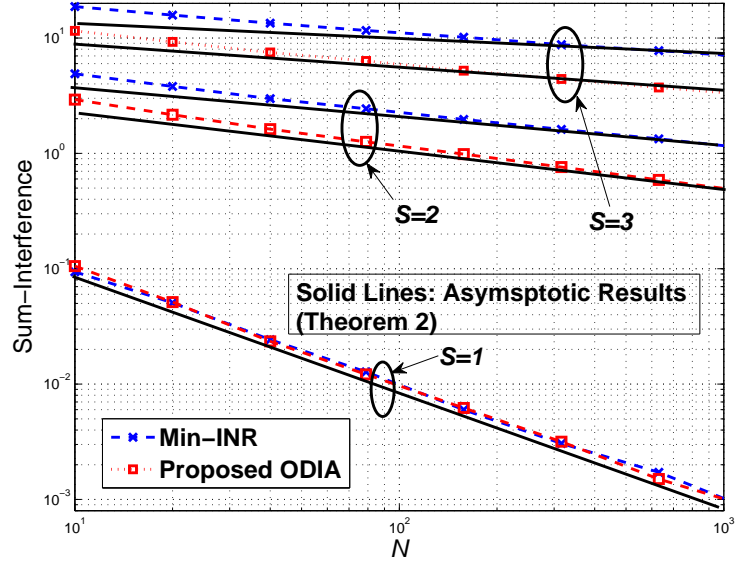}\\
  \caption{Normalized sum-interference vs. $N$ when $K=3$, $M=4$, $L=2$.}\label{fig:Interf_N}
  \end{center}
\end{figure}

\label{line:Fig_varyingN:start}
Fig.~\ref{fig:rate_SNR_varying_nf_N} shows the sum-rate vs. SNR
when $K=2$, $M=3$, $L=2$, and $S=2$. Thus, the total achievable
DoF is $KS=4$. Here, to comply with Theorems \ref{theorem:DoF} and
\ref{th:codebook}, $N$ and $n_f$ are assumed to scale with respect
to SNR as $N=\textrm{SNR}^{(K-1)S-L+1} = \textrm{SNR}^{1}$ and
$n_f = \log_2 \textrm{SNR}$, respectively.
\label{line:Fig_varyingN2:start} For an upper bound, the
genie-aided interference-free ODIA scheme is plotted as
`Interference-Free' in which both the intra- and inter-cell
interference was removed in the achievable rate calculation of the
ODIA scheme.\label{line:Fig_varyingN2:end} It is seen that the
proposed ODIA achieves the target DoF of 4 with
$N=\textrm{SNR}^{(K-1)-L+1}$, which again proves Theorem
\ref{theorem:DoF}. In addition, the ODIA with limited
feedback~(ODIA-LF) also achieves the target DoF of 4 for both
random and Grassmannian codebooks with $n_f = \log_2
(\textrm{SNR})$, which verifies Theorem \ref{th:codebook}. The
Max-SNR scheme achieves zero DoF, since the interference is not
suppressed at all for increasing SNR. The Min-INR scheme cannot
achieve the target DoF, since the user scaling is not fast enough
to satisfy $N=\textrm{SNR}^{KS-L}=\textrm{SNR}^{2}$ (See Section
\ref{subsec:DoF_comparison}).\label{line:Fig_varyingN:end}

\begin{figure}
\begin{center}
  \includegraphics[width=0.65\textwidth]{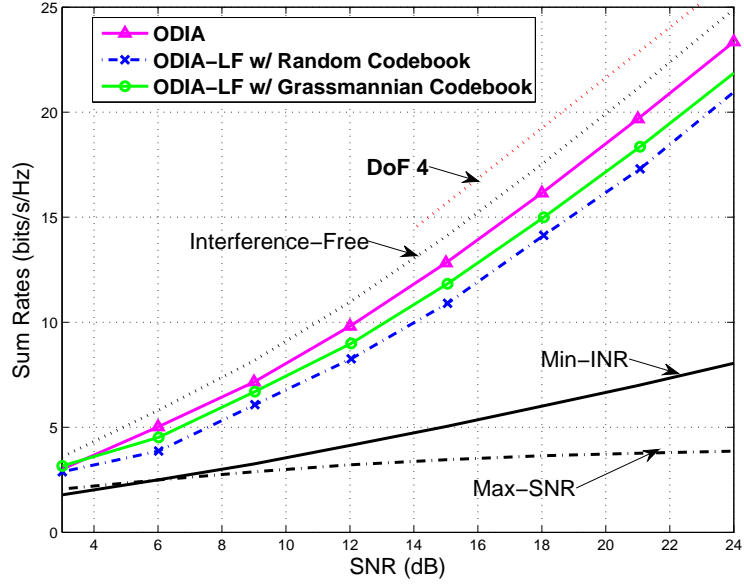}\\
  \caption{Sum-rates versus SNR when $K=2$, $M=3$, $L=2$, $S=2$. The total achievable DoF for all cells is $KS=4$. $N$ and $n_f$ change according to the SNR as $N=\textrm{SNR}^{(K-1)-L+1} = \textrm{SNR}^{1}$ and $n_f = \log_2 \textrm{SNR}$, respectively.}\label{fig:rate_SNR_varying_nf_N}
  \end{center}
\end{figure}

To evaluate the sum-rates of the SE-ODIA, the parameters $\eta_I$,
$\eta_D$, and $\alpha$ need to be optimized for the SE-ODIA. Fig.
\ref{fig:rates_eta} shows the sum-rate performance of the proposed
SE-ODIA for varying $\eta_I$ or $\eta_D$ with two different
$\alpha$ values when $K=3$, $M=4$, $L=2$, $S=2$, and $N=20$. To
obtain the sum-rate according to $\eta_I$, $\eta_D$ was fixed to
$1$. Similarly, for the sum-rate according to $\eta_D$, $\eta_I$
was fixed to $1$. If $\eta_I$ is too small, then there may not be
eligible users that satisfy the conditions $\mathsf{C}_1$ and
$\mathsf{C}_2$ in (\ref{eq:C}). Thus, \textit{scheduling
outage}~\footnote{It indicates the situation that there are no
users who are eligible for scheduling.} can occur frequently and
the achievable sum-rate becomes low. On the other hand, if
$\eta_I$ is too large, then the received interference at users may
not be sufficiently suppressed. Thus, the achievable sum-rate
converges to that of the system without interference suppression.
Similarly, if $\eta_D$ is too large, then the scheduling outage
occurs; and if $\eta_D$ is too small, then desired channel gains
cannot be improved. The orthogonality parameter $\alpha$ plays a
similar role; if $\alpha$ is too small, the cardinality of the
user pool $|\mathcal{N}_s|$ often becomes smaller than $S$, and
scheduling outage happens frequently. If $\alpha$ is too large,
then the orthogonality of the effective channel vectors of the
selected users is not taken into account for scheduling. In short,
the parameters $\eta_I$, $\eta_D$, and $\alpha$ need to be
carefully chosen to improve the performance of the proposed
SE-ODIA. In subsequent sum-rate simulations, proper sets of
$\eta_I$, $\eta_D$, and $\alpha$ were numerically found for
various $N$ and SNR values and applied to the SE-ODIA.
\begin{figure}[t]
\begin{center}
  \includegraphics[width=0.63\textwidth, angle=-0]{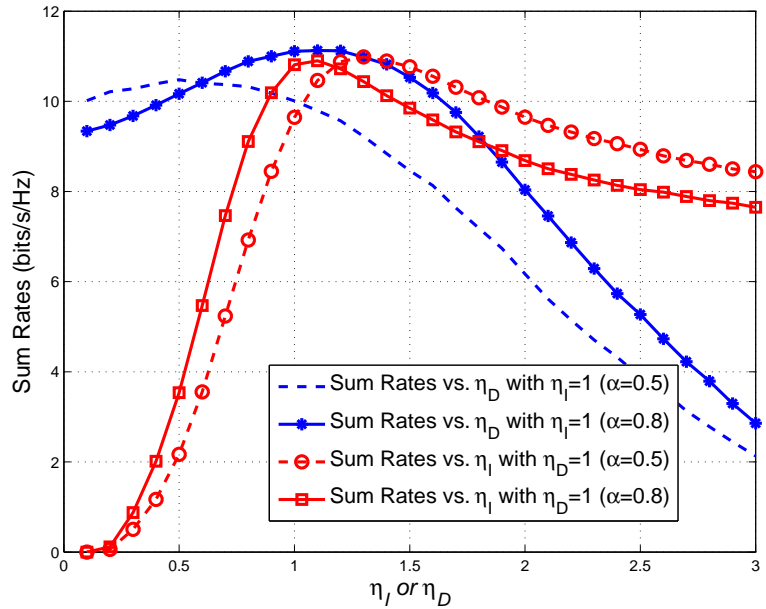}\\
  \caption{Sum-rates of SE-ODIA vs. $\eta_D$ or $\eta_I$ when $K=3$, $M=4$, $L=2$, $S=2$, and $N=20$.}\label{fig:rates_eta}
  \end{center}
\end{figure}
For instance, optimal $(\eta_I, \eta_D, \alpha)$ values that
maximize the sum-rate for a few cases are provided in Table
\ref{table:param}. It is seen that in the noise-limited low SNR
regime, large $\eta_D$ helps, whereas in the interference-limited
high SNR regime, small $\eta_I$ improves the sum-rate. On the
other hand, as $N$ increases, interference can be suppressed by
choosing smaller $\eta_I$ values.
\begin{table} \caption{Optimized parameters $(\eta_I, \eta_D, \alpha)$ for different SNRs and $N$ values}  \label{table:param}
\begin{center}
\begin{tabular}{|c|c|c|}
  \hline
   & $N$=20 & $N$=50 \\ \hline
  SNR=3dB & (2.5, 2.5, 0.8) & (2, 2.5, 0.8) \\ \hline
  SNR=21dB &  (1.5, 2, 0.8) & (1, 2, 0.8) \\
  \hline
\end{tabular}
\end{center}
\end{table}

Fig. \ref{fig:rates_SNR_N20} shows the sum-rates for varying SNR
values when $K=3$, $M=4$, $L=2$, $S=2$, and $N=20$.
 In the noise-limited low SNR regime, the sum-rate of the min-INR scheme is even lower than that of the max-SNR scheme, because $N$ is not large enough to suppress both intra- and inter-cell interference. For comparison, the sum-rate maximizing iterative transceiver design of \cite{Q_Shi11_TSP} is also evaluated allowing one iteration between the BSs and users, i.e., the users feed back their receive beamforming vectors and BSs update their precoding matrices once. Even with one iteration, since each user needs to feed back the information of the receive beamformer to all the BSs in the network, the amount of the feedback is $K$ times more than in the proposed scheme. In addition, because \cite{Q_Shi11_TSP} does not include any consideration of user scheduling, which is in general difficult to be separated from the precoding matrix design, we applied the conventional max-SNR and max-SINR scheduling schemes for the scheme of \cite{Q_Shi11_TSP}, which are labeled by `Max-Sum-Rate w/ Max-SNR Scheduling' and `Max-Sum-Rate w/ Max-SINR Scheduling,' respectively. The precoding matrix was fixed to be the one achieving the max-SNR in the scheduling metric calculation of \cite{Q_Shi11_TSP}, e.g., the scheduling metric for the max-SNR scheme is given by $\mathsf{SNR}\cdot {\lambda_i^{[i,j]}}^2$, where $\lambda_i^{[i,j]}$ is the largest singular value of $\mathbf{H}_i^{[i,j]}$.

 It is seen from the figure that the proposed ODIA outperforms the conventional schemes for SNRs larger than 3dB due to the combined effort of 1) transmit beamforming perfectly eliminating intra-cell interference and 2) receive beamforming effectively reducing inter-cell interference. In particular, the proposed ODIA shows higher sum-rate than the iterative transceiver design even with $K$ times less feedback due to the separate joint optimization of the precoding matrix design and user scheduling.

 The sum-rate performance of the ODIA-LF improves as $n_f$ increases as expected. In practice, $n_f=6$ exhibits a good compromise between the number of feedback bits and sum-rate performance for the codebook dimension of 2 (i.e., $S=2$).
On the other hand, the proposed SE-ODIA achieves higher sum-rates
than the others including the ODIA for all SNR regime, because the
SE-ODIA improves desired channel gains and suppresses interference
simultaneously. Note however that the SE-ODIA includes the
optimization on the parameters for given SNR and $N$ and requires
the user scheduling method based on perfect CSI feedback, which
demands higher computational complexity than the user scheduling
of the ODIA.


\begin{figure}
\begin{center}
\includegraphics[width=0.64\textwidth, angle=0]{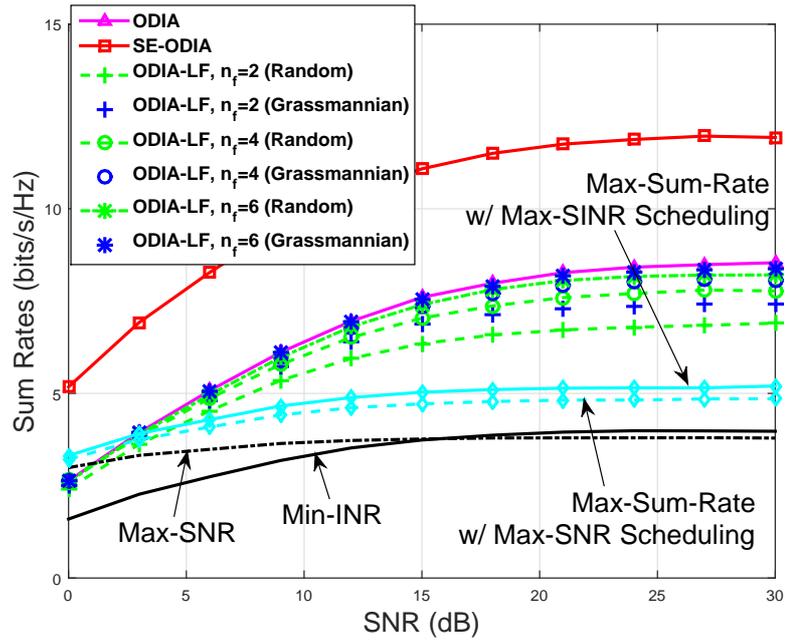}
  \caption{Sum-rates versus SNR when $K=3$, $M=4$, $L=2$, $S=2$, and $N=20$.} \label{fig:rates_SNR_N20}
\end{center}
\end{figure}

Fig. \ref{fig:rates_N_linear} shows the sum-rate performance of
the proposed ODIA schemes for varying number of users per cell,
$N$, when $K=3$, $M=4$, $L=2$, $S=2$, and SNR=$20$dB.
  For limited feedback, the Grassmannian codebook was employed. The sum-rates of the proposed ODIA schemes increase faster than the two conventional schemes, which implies that the user scaling conditions of the proposed ODIA schemes required for a given DoF or MUD gain are lowered than the conventional schemes, as shown in Theorems \ref{theorem:DoF} and \ref{theorem:MUD}.
\begin{figure}
\begin{center}
  \includegraphics[width=0.64\textwidth, angle=-0]{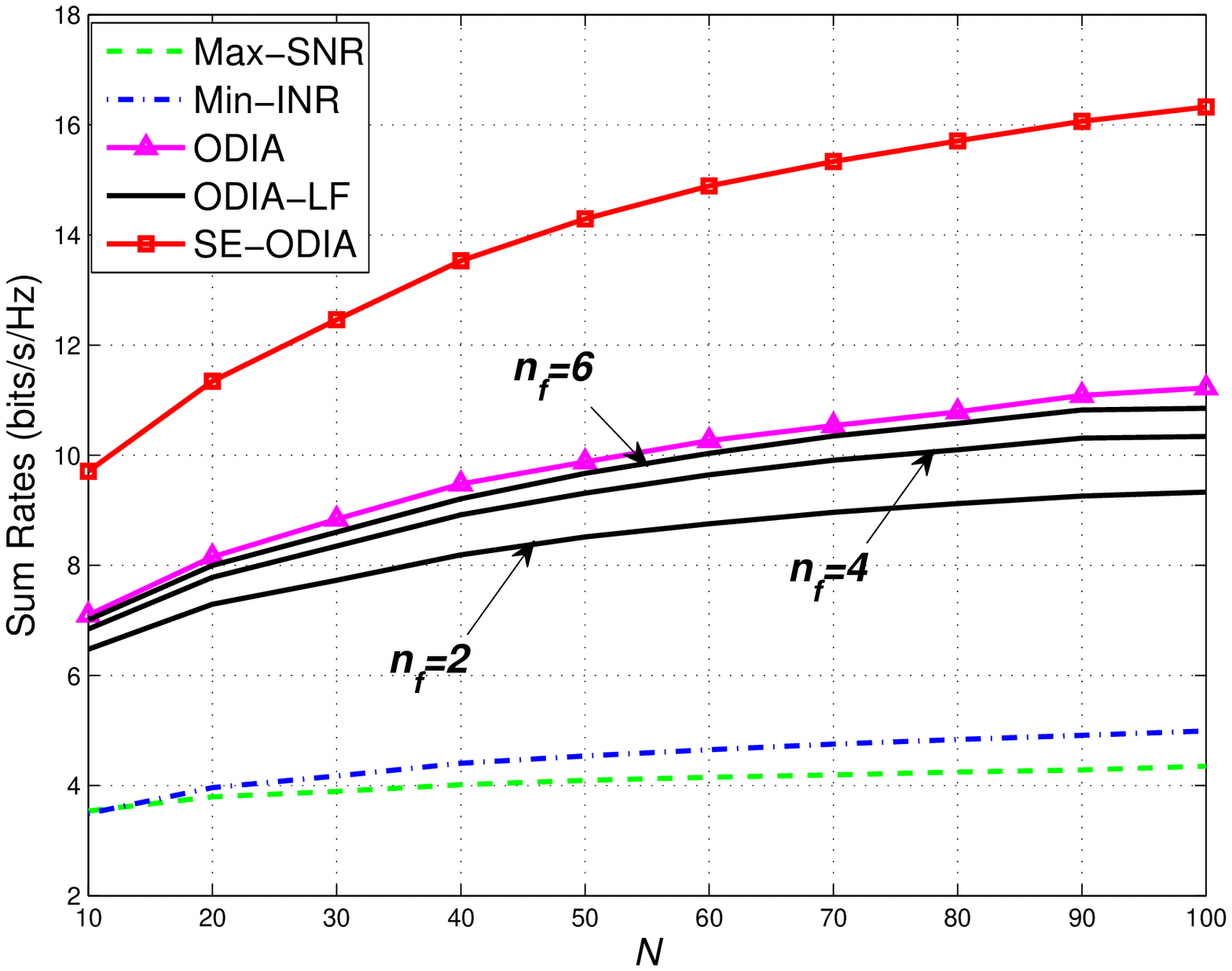}\\
  \caption{Sum-rates vs. $N$ when $K=3$, $M=4$, $L=2$, $S=2$, and SNR=20dB.}\label{fig:rates_N_linear}
  \end{center}
\end{figure}

\section{Conclusion} \label{SEC:Conc}
In this paper, we proposed an opportunistic downlink interference
alignment (ODIA) which intelligently combines user scheduling,
transmit beamforming, and receive beamforming for multi-cell
downlink networks. In the ODIA, the optimal DoF can be achieved
with more relaxed user scaling condition $N=\left(
\mathsf{SNR}^{(K-1)S-L+1}\right)$. To the best of our knowledge,
this user scaling condition is the best known to date. We also
considered a limited feedback approach for the ODIA, and analyzed
the minimum number of feedback bits required to achieve the same
user scaling condition of the ODIA with perfect feedback. We found
that both Grassmannian and random codebooks yield the same
condition on the number of required feedback bits. Finally, a
spectrally efficient ODIA (SE-ODIA) was proposed to further
improve the sum-rate of the ODIA, in which optimal multiuser
diversity can be achieved even in the presence of inter-cell
interference. Through numerical results, it was shown that the
proposed ODIA schemes significantly outperform the conventional
interference management schemes  in practical environments.

\appendices
\section{Proof of Lemma \ref{lemma:CDF_scaling}}\label{app:lemma2}
Using (\ref{eq:sum_interference_equiv}), $\mathcal{P}$ can be
bounded by
\begin{align}
\label{eq:P_LB0}\mathcal{P} & =\lim_{\textsf{SNR}\rightarrow \infty} \textrm{Pr} \left\{\sum_{i=1}^{K}\sum_{j=1}^{S} \eta^{[i,j]}\mathsf{SNR} \le \epsilon \right\} \\
&\ge \lim_{\textsf{SNR}\rightarrow \infty} \textrm{Pr}
\left\{\eta^{[i,j]}\le \frac{\textsf{SNR}^{-1}\epsilon}{KS^2},
\forall i\in \mathcal{K}, \forall j\in
\mathcal{S}\right\}.\label{eq:P_BF_LB1}
\end{align}
Note that the selected users' $\eta^{[i,j]}$ are the minimum $S$ values out of $N$ i.i.d. random variables. 
Since the CDF of $\eta^{[i,j]}$ is given by (\ref{eq:F_phi}),
(\ref{eq:P_BF_LB1}) can be written by
\begin{align}
\label{eq:P_LB0}\mathcal{P} & \ge\lim_{\textsf{SNR}\rightarrow
\infty} \Bigg[ 1- \sum_{i=1}^{S-1}\left( \begin{array}{c}
 N \\
 i
\end{array}
\right)\bigg(\underbrace{F_{\eta}\left(\frac{\epsilon\textsf{SNR}^{-1}}{KS^2}\right)}_{\triangleq A}\bigg)^{i} \nonumber \\ & \quad \quad \cdot \left(1-F_{\eta}\left(\frac{\epsilon\textsf{SNR}^{-1}}{KS^2}\right)\right)^{N-i} \Bigg]\\
\label{eq:P_LB_last}& \ge\lim_{\textsf{SNR}\rightarrow \infty}
\left[ 1- \sum_{i=1}^{S-1}N^i A^{i}\left(1-A\right)^{-i}
\left(1-A\right)^{N}\right],
\end{align}
where
\begin{align}\label{eq:P_LB_lastterm}
\left(1-A\right)^{N}& \!\!\! = \!\! \bigg( 1- c_0 \left(
\frac{\epsilon}{KS^2}\right)^{(K-1)S-L+1} \!\!\! \cdot
\mathsf{SNR}^{-((K-1)S-L+1)} \nonumber \\ & -
\Omega_{\textsf{SNR}}\left(\mathsf{SNR}^{-((K-1)S-L)}\bigg)
\right)^{N}.
\end{align}
Here, $f(x) = \Omega_x\left( g(x) \right)$ means $\lim_{x
\rightarrow \infty} \left|\frac{f(x)}{g(x)}\right|>0$. Thus,
$\left(1-A\right)^{N}$ tends to 0 (exponentially) if and only if
$N$ scales faster than $\mathsf{SNR}^{(K-1)S-L+1}$. Now, inserting
$N=\omega\left( \mathsf{SNR}^{(K-1)S-L+1}\right)$ to
(\ref{eq:P_LB_last}) yields $\mathcal{P}$ tending to 1 for
increasing SNR, because for given $i$, $\left(1-A\right)^{N}$
vanishes exponentially.

\section{Proof of Theorem \ref{th:codebook}} \label{app:th_codebook}
i) Grassmannian codebook \\
For the Grassmannian codebook, the chordal distance between any
two codewords is the same, i.e., $\sqrt{1-\left|
\mathbf{c}_i^{\mathsf{H}} \mathbf{c}_j\right|^2} = d,$, $\forall i
\neq j$. The Rankin, Gilbert-Varshamov, and Hamming bounds on the
chordal distance give us \cite{A_Barg02_TIT,J_Conway96_EM,
W_Dai08_TIT}
\begin{equation} \label{eq:dmin_bound}
{d^{[i,j]}}^2 \le  \min \left\{\frac{1}{2},
\frac{(S-1)N_f}{2S(N_f-1)}, \left(
\frac{1}{N_f}\right)^{1/(S-1)}\right\}.
\end{equation}
The bound in (\ref{eq:dmin_bound}) is reduced to the third bound
as $N_f$ increases, thus providing arbitrarily tight upper-bound
on ${d^{[i,j]}}^2$. Thus,  the first term of (\ref{eq:Delta})
remains constant if
\begin{align}
\left(d^{\max}_i\right)^2\delta_1\cdot \mathsf{SNR} & \le \left(
\frac{1}{N_f}\right)^{1/(S-1)}\delta_1\cdot \mathsf{SNR}\le
\epsilon'.
\end{align}
This is reduced to $N_f^{-1/(S-1)} \le \epsilon' \delta_1^{-1}
\mathsf{SNR}^{-1}$,
 or equivalently (\ref{eq:nf_cond0}).
Now, if (\ref{eq:nf_cond0}) holds true, $d_i^{\max}$ tends to be
arbitrarily small as SNR increases, and thus the second term of
(\ref{eq:Delta}) is dominated by the first term.
Therefore, if $n_f$ scales with respect to $\log_2(\mathsf{SNR})$
as (\ref{eq:nf_cond0}), the residual intra-cell interference
$\Delta^{[i,j]}$ remains constant.

ii) Random codebook \\
In a random codebook, each codeword $\mathbf{c}_k$ is chosen
isotropically and independently from the $L$-dimensional hyper
sphere, and thus the maximum chordal distance of a random codebook
is unbounded. Since ${d^{[i,j]}}^2$ is the minimum of $N_f$
chordal distances resulting from $N_f$ independent codewords, the
CDF of ${d^{[i,j]}}^2$ is given by
\cite{C_AuYeung07_TWC,N_Jindal06_TIT}
\begin{equation} \label{eq:F_d_def}
F_d(z) \triangleq \textrm{Pr}\left\{{d^{[i,j]}}^2\le z\right\} =
1-\left(1-z^{S-1}\right)^{N_f}.
\end{equation}

From (\ref{eq:Delta}), the second term of (\ref{eq:P_def3}) can be
bounded by
\begin{align} \label{eq:PPP}
&\textrm{Pr} \Bigg\{\Delta^{[i,j]} \le \epsilon', \forall i\in \mathcal{K}, j\in \mathcal{S}  \Bigg\} \nonumber \\
& \ge \textrm{Pr}\left\{\left(d^{\max}_i\right)^2\delta_1\cdot
\mathsf{SNR}\le \epsilon', \forall i\in \mathcal{K}\right\}
\nonumber \\  & \quad \quad \cdot \textrm{Pr} \left\{
\sum_{k=2}^{\infty}  \left(d^{\max}_i\right)^{2k}\delta_k \cdot
\mathsf{SNR} \le \epsilon', \forall i\in \mathcal{K} \right\}.
\end{align}
Subsequently, we have
\begin{align} \label{eq:RV_P1}
&\textrm{Pr} \!\! \left\{\left(d^{\max}_i\right)^2\delta_1\cdot
\mathsf{SNR}\le \epsilon'\right\} \!\! = \!\! \prod_{k=1}^{S} \!\!
\textrm{Pr} \!  \left\{\left(d^{[k,i]}\right)^2 \!\! \delta_1 \!
\cdot \! \mathsf{SNR}\le \epsilon'\right\},
\end{align}
which follows from the fact that $d^{[k,i]}$ and $d^{[m,i]}$ are
independent for $k\neq m$. From (\ref{eq:F_d_def}) we have
\begin{align} \label{eq:RV_P1_2}
&\textrm{Pr}\left\{\left(d^{[k,i]}\right)^2\delta_1\cdot
\mathsf{SNR}\le \epsilon'\right\}  \nonumber \\ &=
1-\left(1-{\epsilon'}^{S-1}\delta_1^{-S+1}\left(
\mathsf{SNR}\right)^{-(S-1)}\right)^{N_f}.
\end{align}
Therefore, $\lim_{\mathsf{SNR}\rightarrow \infty}
\textrm{Pr}\left\{\left(d^{\max}_i\right)^2\delta_1\cdot
\mathsf{SNR}\le \epsilon'\right\}=1$  if and only if $N_f = \omega
\left( \mathsf{SNR}^{S-1}\right)$,
or equivalently (\ref{eq:nf_cond0}).
Now, if (\ref{eq:nf_cond0}) holds true, $d_{i}^{\max}$ tends to
arbitrarily small with high probability as SNR increases.
Therefore, the second term of (\ref{eq:Delta}) is dominated by the
first term, and hence $\textrm{Pr} \left\{\Delta^{[i,j]} \le
\epsilon', \forall i\in \mathcal{K}, j\in \mathcal{S}  \right\}$
in (\ref{eq:PPP}) tends to 1.

\section{Proof of Lemma \ref{lemma:N_card}} \label{app:Ns_cardinality}
Let us define the set $\Pi_s$ by
\begin{align}
& \Pi_s \triangleq  \nonumber \\ & \left\{ \mathbf{h}\in
\mathbb{C}^{S \times 1}: \frac{{\mathbf{h}}^{\mathsf{H}}
\mathbf{v}}{\| \mathbf{h}\| \|\mathbf{v}\|}  <\alpha, \forall
\mathbf{v}\in \textrm{span}\left( \mathbf{b}_{1}^{[i]}, \ldots,
\mathbf{b}_{s-1}^{[i]}\right) \right\}.
\end{align}
Since the $s$-th user pool is determined only by checking the
orthogonality to the chosen users' channel vectors, for
arbitrarily large $N$, we have the followings by the law of large
numbers: \pagebreak[0]
\begin{align}
\label{eq:N_card1} |\mathcal{N}_s| & \! \approx \! N \! \cdot \!
\textrm{Pr} \! \left\{ \mathbf{h}\in \mathbb{C}^{S \times 1}:
\frac{{\mathbf{h}}^{\mathsf{H}} {\mathbf{b}_{s}^{[i]}}'}{\|
\mathbf{h}\| \|\mathbf{b}_{s'}^{[i]}\|}  <\alpha, s'=1, \ldots,
s-1  \right\} \\ \pagebreak[0]
& \ge N \cdot \textrm{Pr}\left\{ \mathbf{h}\in \mathbb{C}^{S \times 1}: \mathbf{h} \in \Pi_s \right\} \\
&= N \cdot I_{\alpha^2}(s-1, S-s+1) \\ \pagebreak[0]
\label{eq:N_card2}&\ge N \cdot \alpha^{2(S-1)}, \pagebreak[0]
\end{align}
where $I_{\alpha^2}$ is the regularized incomplete beta function
(See \cite[Lemma 3]{T_Yoo06_JSAC}), and (\ref{eq:N_card2}) follows
from $I_{\alpha^2}(s-1, S-s+1) \ge I_{\alpha^2}(S-1, 1) =
\alpha^{2(S-1)}$.

\section{Proof of (\ref{eq:ps_conv})} \label{app:MUD}
Since $\left\|\mathbf{b}_{s}^{[i]}\right\|^2$ is a Chi-squared
random variable with $2(S-s+1)$ degrees-of-freedom, for
$\eta_D>2$, we have
\begin{align}
\textrm{Pr} \left\{ \mathsf{C}_2  \right\} 
& = 1-\frac{\gamma((S-s+1), \eta_D/2)}{\Gamma(S-s+1 )}  \\ &= \frac{\Gamma((S-s+1), \eta_D/2)}{\Gamma(S-s+1 )}\\
& = \sum_{m=0}^{S-s} e^{-(\eta_D/2)} \frac{{(\eta_D/2)}^{m}}{m!} \\
& = \frac{e^{-(\eta_D/2)}\cdot {(\eta_D/2)}^{S-s}}{(S-s)!}\left(1+ O\left({(\eta_D/2)}^{-1}\right)\right)\\
\label{eq:C2_final}&\ge \frac{e^{-(\eta_D/2)}}{(S-s)!},
\end{align}
where $\Gamma(s,x) = \int_{x}^{\infty} t^{s-1}e^{-t}dt$ is the
upper incomplete gamma function and $\gamma(s,x)=
\int_{0}^{x}t^{s-1}e^{-t}dt$ is the lower incomplete gamma
function.

Note that from the CDF of $\eta^{[i,j]}$ (See \cite[Lemma
1]{H_Yang13_TWC}), $\textrm{Pr} \left\{ \eta^{[i,j]} \le
\eta_I\right\} = c_0 \eta_I^{\tau} + o(\eta_I^{\tau})$, where
$\tau = (K-1)S-L+1$. Thus, from (\ref{eq:eta_D_choice}),
(\ref{eq:eta_I_choice}), and (\ref{eq:C2_final}),  (\ref{eq:p_s2})
can be bounded by
\begin{align} \label{eq:p_s_final}
\mathsf{p}_s & \gtrsim 1- \Bigg( 1- \left( c_0(\epsilon_I)^{\tau} {\mathsf{SNR}}^{-\tau} + \Omega\left( {\mathsf{SNR}}^{-(\tau-1)}\right)\right)\nonumber \\
& \hspace{100pt} \times
\frac{N^{-(\epsilon_D/2)}}{(S-s)!}\Bigg)^{N \cdot
\alpha^{2(S-1)}}.
\end{align}
The right-hand side of (\ref{eq:p_s_final}) converges to 1 for
increasing SNR if and only if
\begin{align}\label{eq:scaling_MUD}
\lim_{\mathsf{SNR}\rightarrow \infty} & \left(N \cdot
\alpha^{2(S-1)}\right) \cdot \left( c_0 (\epsilon_I)^{\tau}
{\mathsf{SNR}}^{-\tau} + \Omega\left(
{\mathsf{SNR}}^{-(\tau-1)}\right)\right) \nonumber \\ & \cdot
\frac{N^{-(\epsilon_D/2)}}{(S-s)!} = \infty.
\end{align}
Since the left-hand side of (\ref{eq:scaling_MUD}) can be written
by $ \tilde{c}_0\frac{N^{1-(\epsilon_D/2)}}{\mathsf{SNR}^{\tau}} +
\tilde{c}_1\frac{N^{1-(\epsilon_D/2)}}{o\left(\mathsf{SNR}^{\tau}\right)}$,
where $\tilde{c}_0$ and $\tilde{c}_1$ are positive constants
independent of SNR and $N$, it tends to infinity for increasing
SNR, and thereby  $\mathsf{p}_s$ tends to 1 if and only if $N =
\omega \left(
\mathsf{SNR}^{\frac{(K-1)S-L+1}{1-(\epsilon_D/2)}}\right)$.




\begin{thebibliography}{10}
\providecommand{\url}[1]{#1} \csname url@samestyle\endcsname
\providecommand{\newblock}{\relax}
\providecommand{\bibinfo}[2]{#2}
\providecommand{\BIBentrySTDinterwordspacing}{\spaceskip=0pt\relax}
\providecommand{\BIBentryALTinterwordstretchfactor}{4}
\providecommand{\BIBentryALTinterwordspacing}{\spaceskip=\fontdimen2\font
plus \BIBentryALTinterwordstretchfactor\fontdimen3\font minus
  \fontdimen4\font\relax}
\providecommand{\BIBforeignlanguage}[2]{{%
\expandafter\ifx\csname l@#1\endcsname\relax
\typeout{** WARNING: IEEEtran.bst: No hyphenation pattern has been}%
\typeout{** loaded for the language `#1'. Using the pattern for}%
\typeout{** the default language instead.}%
\else \language=\csname l@#1\endcsname \fi #2}}
\providecommand{\BIBdecl}{\relax} \BIBdecl

\bibitem{V_Cadambe08_TIT}
V.~R. Cadambe and S.~A. Jafar, ``Interference alignment and
degrees of freedom
  of the \uppercase{K}-user interference channel,'' \emph{IEEE Trans. Inf.
  Theory}, vol.~54, no.~8, pp. 3425--3441, Aug. 2008.

\bibitem{K_Gomadam11_TIT}
K.~Gomadam, V.~R. Cadambe, and S.~A. Jafar, ``A distributed
numerical approach
  to interference alignment and applications to wireless interference
  networks,'' \emph{IEEE Trans. Inf. Theory}, vol.~57, no.~6, pp. 3309--3322,
  June 2011.

\bibitem{T_Gou10_TIT}
T.~Gou and S.~A. Jafar, ``Degrees of freedom of the \uppercase{K}
user
  \uppercase{M X N MIMO} interference channel,'' \emph{IEEE Trans. Inf.
  Theory}, vol.~56, no.~12, pp. 6040--6057, Dec. 2010.

\bibitem{C_Suh11_TC}
C.~Suh, M.~Ho, and D.~Tse, ``Downlink interference alignment,''
\emph{IEEE
  Trans. Commun.}, vol.~59, no.~9, pp. 2616--2626, Sept. 2011.

\bibitem{C_Suh08_Allerton}
C.~Suh and D.~Tse, ``Interference alignment for cellular
networks,'' in
  \emph{Proc. 46th Annual Allerton Conf. Communication, Control, and
  Computing}, Urbana-Champaign, IL, Sept. 2008, pp. 1037 -- 1044.

\bibitem{R_Knopp95_ICC}
R.~Knopp and P.~Humblet, ``Information capacity and power control
in single
  cell multiuser communications,'' in \emph{Proc. Int'l Conf. Commun. (ICC)},
  Seattle, WA, June 1995, pp. 331--335.

\bibitem{P_Viswanath02_TIT}
P.~Viswanath, D.~N.~C. Tse, and R.~Laroia, ``Opportunistic
beamforming using
  dumb antennas,'' \emph{IEEE Trans. Inf. Theory}, vol.~48, no.~6, pp.
  1277--1294, Aug. 2002.

\bibitem{M_Sharif05_TIT}
M.~Sharif and B.~Hassibi, ``On the capacity of \uppercase{MIMO}
broadcast
  channels with partial side information,'' \emph{IEEE Trans. Inf. Theory},
  vol.~51, no.~2, pp. 506--522, Feb. 2005.

\bibitem{W_Shin14_TIT}
W.-Y. Shin, S.-Y. Chung, and Y.~H. Lee, ``Parallel opportunistic
routing in
  wireless networks,'' \emph{IEEE Trans. Inf. Theory}, to appear.

\bibitem{T_Ban09_TWC}
T.~W. Ban, W.~Choi, B.~C. Jung, and D.~K. Sung, ``Multi-user
diversity in a
  spectrum sharing system,'' \emph{IEEE Trans. Wireless Commun.}, vol.~8,
  no.~1, pp. 102--106, Jan. 2009.

\bibitem{W_Shin12_TC}
W.~Y. Sin, D.~Park, and B.~C. Jung, ``Can one achieve multiuser
diversity in
  uplink multi-cell networks?'' \emph{IEEE Trans. Commun.}, vol.~60, no.~12,
  pp. 3535--3540, Dec. 2012.

\bibitem{B_Jung11_CL}
B.~C. Jung and W.-Y. Shin, ``Opportunistic interference alignment
for
  interference-limited cellular \uppercase{TDD} uplink,'' \emph{IEEE Commun.
  Lett.}, vol.~15, no.~2, pp. 148--150, Feb. 2011.

\bibitem{B_Jung11_TC}
B.~C. Jung, D.~Park, and W.-Y. Shin, ``Opportunistic interference
mitigation
  achieves optimal degrees-of-freedom in wireless multi-cell uplink networks,''
  \emph{IEEE Trans. Commun.}, vol.~60, no.~7, pp. 1935--1944, July 2012.

\bibitem{H_Yang13_TWC}
H.~J. Yang, W.-Y. Shin, B.~C. Jung, and A.~Paulraj,
``Opportunistic
  interference alignment for {MIMO} interfering multiple access channels,''
  \emph{IEEE Trans. Wireless Commun.}, vol.~12, no.~5, pp. 2180--2192, May
  2013.

\bibitem{W_Shin12_IEICE}
W.-Y. Shin and B.~C. Jung, ``Network coordinated opportunistic
beamforming in
  downlink cellular networks,,'' \emph{IEICE Trans. Commun.}, vol. E95-B,
  no.~4, pp. 1393--1396, Apr. 2012.

\bibitem{J_Jose12_Allerton}
J.~Jose, S.~Subramanian, X.~Wu, and J.~Li, ``Opportunistic
interference
  alignment in cellular downlink,'' in \emph{50th Annual Allerton Conference on
  Communication, Control, and Computing (Allerton)}, 2012, pp. 1529--1545.

\bibitem{J_Lee13_TWC}
J.~H. Lee and W.~Choi, ``On the achievable dof and user scaling
law of
  opportunistic interference alignment in 3-transmitter \uppercase{MIMO}
  interference channels,'' \emph{IEEE Trans. Wireless Commun.}, vol.~12, no.~6,
  pp. 2743--2753, Jun. 2013.

\bibitem{H_Nguyen13_TSP}
H.~D. Nguyen, R.~Zhang, and H.~T. Hui, ``Multi-cell random
beamforming:
  Achievable rate and degrees-of-freedom region,'' \emph{IEEE Trans. Signal
  Process.}, vol.~61, no.~14, pp. 3532--3544, July 2013.

\bibitem{H_Nguyen13_arXiv}
------, ``Effect of receive spatial diversity on the degrees-of-freedom region
  in multi-cell random beamforming,'' \emph{IEEE Trans. Wireless Commun.},
  submitted, Preprint, [Online]. Available: http://arxiv.org/abs/1303.5947.

\bibitem{J_Lee13_arXiv}
J.~H. Lee, W.~Choi, and B.~D. Rao, ``Multiuser diversity in
interfering
  broadcast channels: Achievable degrees of freedom and user scaling law,''
  \emph{IEEE Trans. Wireless Commun.}, vol.~12, no.~11, pp. 5837--5849, Nov.
  2013.

\bibitem{N_Jindal06_TIT}
N.~Jindal, ``\uppercase{MIMO} broadcast channels with finite-rate
feedback,''
  \emph{IEEE Trans. Inform. Theory}, vol.~52, no.~11, Nov. 2006.

\bibitem{T_Yoo07_JSAC}
T.~Yoo, N.~Jindal, and A.~Goldsmith, ``Multi-antenna downlink
channels with
  limited feedback and user selection,'' \emph{IEEE J. Select. Areas Commun.},
  vol.~25, no.~7, pp. 1478--1491, Sept. 2007.

\bibitem{J_Thukral09_ISIT}
J.~Thukral and H.~B\"{o}lcskei, ``Interference alignment with
limited
  feedback,'' in \emph{Proc. IEEE Int'l Symp. Inf. Theory (ISIT)}, Seoul,
  Korea, July 2009, pp. 1759--1763.

\bibitem{R_Krishnamachari10_ISIT}
R.~T. Krishnamachari and M.~K. Varanasi, ``Interference alignment
under limited
  feedback for \uppercase{MIMO} interference channels,'' in \emph{Proc. IEEE
  Int'l Symp. Inf. Theory (ISIT)}, AUstin, TX, June 2010, pp. 619--623.

\bibitem{S_Pereira07_Asilomar}
S.~Pereira, A.~Paulraj, and G.~Papanicolaou, ``Opportunistic
scheduling for
  multiantenna cellular: Interference limited regime,'' in \emph{Proc. Asilomar
  Conference on Signals, Systems and Computers}, Pacific Grove, CA, Nov. 2007.

\bibitem{S_Jafar08_TIT}
S.~A. Jafar and S.~Shamai~(Shitz), ``Dgrees of freedom region of
the
  \uppercase{MIMO} \uppercase{X} channel,'' \emph{IEEE Trans. Inf. Theory},
  vol.~54, no.~1, pp. 151--170, Jan. 2008.

\bibitem{Q_Shi11_TSP}
Q.~Shi, M.~Razaviyayn, Z.~Q. Luo, and C.~He, ``An iteratively
weighted {MIMO}
  approach to distributed sum-utility maximization for a {MIMO} interfering
  broadcast channel,'' \emph{IEEE Trans. Signal Process.}, vol.~59, no.~9, pp.
  4331--4340, Sep. 2011.

\bibitem{H_Yang14_TSP}
H.~J. Yang, W.-Y. Shin, B.~C. Jung, and A.~Paulraj,
``Opportunistic downlink
  interference alignment,'' \emph{IEEE Trans. Signal Process.}, vol.~62,
  no.~11, pp. 2292--2937, June 2014.

\bibitem{L_Choi04_TWC}
L.~Choi and R.~D. Murch, ``A transmit preprocessing technique for
multiuser
  \uppercase{MIMO} systems using a decomposition approach,'' \emph{IEEE Trans.
  Wireless Commun.}, vol.~3, no.~1, pp. 20--24, Jan. 2004.

\bibitem{D_Love03_TIT}
D.~J. Love and R.~W. Heath, Jr., ``Grassmannian beamforming for
multiple-input
  multple-output wireless systems,'' \emph{IEEE Trans. Inf. Theory}, vol.~49,
  no.~10, pp. 2735--2747, Oct. 2003.

\bibitem{Z_Peng16_TWC}
Z.~Peng, W.~Xu, J.~Zhu, H.~Zhang, and C.~Zhao, ``On performance
and feedback
  strategy of secure multiuser communications with {MMSE} channel estimate,''
  \emph{IEEE Trans. Wireless Commun.}, vol.~15, no.~2, pp. 1602--1616, Feb.
  2016.

\bibitem{C_Au-Yeung09_TWC}
C.~K. Au-Yeung and D.~J. Love, ``Optimization and tradeoff
analysis of two-way
  limited feedback beamforming systems,'' \emph{IEEE Trans. Wireless Commun.},
  vol.~8, no.~5, pp. 2570--2579, May 2009.

\bibitem{B_Khoshnevis11_Thesis}
B.~Khoshnevis, ``Multiple-antenna communications with limited
channel state
  information,'' Ph.D. dissertation, University of Toronto, 2011.

\bibitem{T_Yoo06_JSAC}
T.~Yoo and A.~Goldsmith, ``On the optimality of multi-antenna
broadcast
  scheduling using zero-forcing beamforming,'' \emph{IEEE J. Select. Areas
  Commun.}, vol.~24, pp. 528--541, 2006.

\bibitem{G_Seber03_Book}
G.~A.~F. Seber and J.~L. Alan, \emph{Linear Regression Analysis},
2nd~ed.\hskip
  1em plus 0.5em minus 0.4em\relax Wiley, 2003.

\bibitem{A_Barg02_TIT}
A.~Barg and D.~Y. Nogin, ``Bounds on packings of spheresin the
  \uppercase{G}rassmann manifold,'' \emph{IEEE Trans. Inf. Theory}, vol.~48,
  no.~9, pp. 2450--2454, Sept. 2002.

\bibitem{J_Conway96_EM}
J.~H. Conway, R.~H. Hardin, and N.~J.~A. Sloane, ``Packing lines,
planes, etc.:
  Packings in \uppercase{G}rassmannian spaces,'' \emph{Experimental
  Mathematics}, vol.~5, pp. 139--159, 1996.

\bibitem{W_Dai08_TIT}
W.~Dai, Y.~E. Liu, and B.~Rider, ``Quantization bounds on
\uppercase{G}rassmann
  manifolds and applications to \uppercase{MIMO} communications,'' \emph{IEEE
  Trans. Inf. Theory}, vol.~54, no.~3, pp. 1108--1123, Mar. 2008.

\bibitem{C_AuYeung07_TWC}
C.~K. Au-Yeung and D.~J. Love, ``On the performance of random
vector
  quantization limited feedback beamforming in a \uppercase{MISO} system,''
  \emph{IEEE Trans. Wireless Commun.}, vol.~6, no.~2, pp. 458--462, Feb. 2007.

\end{thebibliography}
\end{document}